\documentclass{elsarticle}
\makeatletter
\def\ps@pprintTitle{%
 \let\@oddhead\@empty
 \let\@evenhead\@empty
 \def\@oddfoot{}%
 \let\@evenfoot\@oddfoot}
\makeatother
\AtBeginDocument{%
  \providecommand\BibTeX{{%
    \normalfont B\kern-0.5em{\scshape i\kern-0.25em b}\kern-0.8em\TeX}}}

\usepackage[utf8]{inputenc}
\usepackage{mathtools}
\usepackage{todonotes}
\usepackage{amsfonts}
\usepackage{graphicx}
\usepackage{algorithm}
\usepackage[]{algpseudocode}
\usepackage{subcaption}
\usepackage{xspace}
\usepackage{mathrsfs}
\usepackage{morefloats}
\usepackage{pifont}
\usepackage{lineno}
\usepackage{textcomp, mathcomp}
\usepackage{array}
\usepackage{comment}
\usepackage{enumerate}
\usepackage{enumitem}
\usepackage{soul}
\usepackage{hyperref}
\usepackage{appendix}

\usepackage{amsthm}

\newtheorem{theorem}{Theorem}[section]
\newtheorem{proposition}{Proposition}[section]
\newtheorem{lemma}{Lemma}[section]

\theoremstyle{definition}
\newtheorem{definition}{Definition}[section]
\newtheorem{remark}{Remark}[section]
\newtheorem{problem}{Problem}

\usepackage[english]{babel}

\newcommand{\sparsity}{\textnormal{sparsity}}
\newcommand{\lightness}{\textnormal{lightness}}
\newcommand{\cost}{\textnormal{cost}}
\newcommand{\val}{\textnormal{value}}
\newcommand{\eps}{\varepsilon}
\newcommand{\R}{\mathbb{R}}
\newcommand{\N}{\mathbb{N}}

\newcommand{\PP}{\mathcal{P}}
\newcommand{\dist}{d}
\newcommand{\In}{\textnormal{In}}
\newcommand{\Out}{\textnormal{Out}}
\newcommand{\OPT}{\textnormal{OPT}}
\newcommand{\ILP}{\textnormal{ILP}}
\newcommand{\diam}{\textnormal{diam}}
\newcommand{\MST}{\textnormal{MST}}
\newcommand{\frag}{\textnormal{frag}}
\newcommand{\ceiling}[1]{\lceil #1 \rceil}
\newcommand{\Prob}{\mathbb{P}}
\newcommand{\E}{\mathbb{E}}

\makeatletter
\def\blfootnote{\xdef\@thefnmark{}\@footnotetext}
\makeatother

\usepackage{bibentry}

\nobibliography*

\usepackage{relsize}

\let\olditemize\itemize
\let\endolditemize\enditemize
\renewenvironment{itemize}{%
    \smaller
    \olditemize
}{%
    \endolditemize
}

\begin{document}

\title{Graph Spanners: A Tutorial Review}

\author[add1]{Reyan~Ahmed}
\author[add3]{Greg~Bodwin}
\author[add1]{Faryad~Darabi~Sahneh}
\author[add2]{Keaton Hamm}
\author[add2]{Mohammad Javad Latifi Jebelli}
\author[add1]{Stephen~Kobourov}
\author[add1]{Richard~Spence}

\address[add1]{Department of Computer Science, University of Arizona}
\address[add2]{Department of Mathematics, University of Arizona}
\address[add3]{Department of Computer Science, Georgia Institute of Technology}

\begin{abstract}
This survey provides a guiding reference to researchers seeking an
overview of the large body of literature about graph spanners. It surveys the current literature covering
various research streams about graph spanners, such as different formulations, sparsity and lightness results, computational complexity, dynamic algorithms, and applications. As an additional contribution, we offer a list of open problems on graph spanners.
\end{abstract}

\maketitle

\blfootnote{This work is supported in part by NSF grants CCF-1740858, CCF-0721503 and DMS-1839274.}

\tableofcontents

\section{Introduction}
Given a graph $G$, a \emph{graph spanner} (or simply \emph{spanner}) is a subgraph which preserves lengths of shortest paths in $G$ up to some amount of distortion or error, typically multiplicative and/or additive. Spanners were introduced by Peleg and Sch\"{a}ffer~\cite{peleg1989graph}, and computing sparse or low-weight spanners has theoretical and practical applications in various network design problems, including graph compression and those in distributed computing and communication networks~\cite{awerbuch1985complexity,Bhatt1986Optimal}. 
This document provides a survey on the current literature of graph spanners, including different problem variants, hardness results, common techniques, and related open problems. {\color{black}In particular, geometric spanners are de-emphasized in this survey; we instead refer the reader to~\cite{Gudmundsson2008, Narasimhan:2007:GSN:1208237}.} The focus is on a method-based tutorial style which will allow those unfamiliar with the area to survey the problems considered and typical techniques used in the literature to compute spanners.

 \subsection{Organization and Layout}
Throughout this survey, we use an annotated bibliography style in which important references are included within sections or subsections throughout the main text; this is done to enhance the readability of the survey so that the reader can find references quickly without going back and forth to the end of the survey.  Section~\ref{SEC:notation} introduces notation and basic preliminaries which will be used throughout the survey. Section~\ref{SEC:Spanners} gives a precise definition of various spanners considered in the literature as well as classical results and relationships to well-known concepts such as minimum spanning trees and Erd\H{o}s' Girth Conjecture.  Section~\ref{SEC:Complexity} details the complexity and hardness of approximation for various spanner problems. Sections~\ref{SEC:Greedy}--\ref{SEC:ILP}  comprise the bulk of the survey and describe the main techniques brought to bear on spanner problems to date including greedy algorithms, clustering and path buying algorithms, probabilistic constructions, ILP formulations, and LP-based algorithms.  The emphasis of these sections is on illustrating the methods and inclusion of proofs to give the main ideas, as well as on posing open problems related to the works discussed. Section~\ref{SEC:Distributed} concerns distributed algorithms. Section~\ref{SEC:TypesOfSpanners} briefly shows other kinds of spanners considered in the literature, while Section~\ref{SEC:ChangingGraphs} discusses spanners when the underlying graph is allowed to change.  Section~\ref{SEC:RestrictedGraphs} discusses spanners for restricted classes of graphs, and Section~\ref{SEC:Applications} ends the paper with some of the numerous applications of graph spanners.  At the end of the paper, there are several summary tables for guarantees on the size and weight of spanners based on the type of problem.

\subsection{{\color{black}Notation and Preliminaries}} \label{SEC:notation}
Graphs are denoted $G = (V,E)$ containing $|V|=n$ vertices (or nodes) and $|E|=m$ edges, and are assumed undirected and connected unless stated otherwise. {\color{black}We use $uv$ to denote the undirected edge $\{u,v\}$, and $(u,v)$ to denote the directed edge from $u$ to $v$.} We denote by $w_e=w_{uv}$ the weight of edge $e=uv$ (for weighted graphs), and denote by $d_G(u,v)$ the weight of a minimum-weight $u$-$v$ path in $G$ (or, the number of edges in a shortest path if $G$ is unweighted). Given a subset ${\color{black}\bar{E}\subseteq E}$ of edges, its weight is denoted by $W(\bar{E}):=\sum_{e\in \bar{E}}w_e$. In particular, we denote by $W(\MST(G))$ the weight of a minimum spanning tree (\MST) of $G$. 
Given two sets $A,B\subseteq V$, we define the distance from $A$ to $B$ in $G$ as $\textnormal{dist}_{G}(A,B) = \min\{d_G(u,v):u\in A, v\in B\}$; if $A=\{v\}$ for some vertex $v$ (resp. $B=\{v\}$), we will use $\textnormal{dist}_{G}(v,B)$ (resp. $\textnormal{dist}_G(A,v)$).  The diameter of $G$ is denoted by $\diam(G)=\max_{u,v\in V}d_G(u,v)$. {\color{black}For degree-bounded graphs, we typically denote by $\Delta$ the maximum degree of any vertex in the graph.}

Recall that for functions $f,g: \mathbb{R} \to \mathbb{R}_{\ge 0}$, we say that $f(n)=O(g(n))$ if there exist constants $C, n_0 > 0$ such that $f(n) \le Cg(n)$ for $n \ge n_0$. Further, $f=\widetilde{O}(g)$ if {\color{black}$f=O(g\,\textnormal{polylog}\,n)$}\footnote{$\text{polylog\,}n$ denotes a polynomial in terms of $\log n$, i.e., $a_k (\log n)^k + \ldots + a_1 \log n + a_0$.}.  Additionally, $f=O_\eps(g)$ if $f=O(\text{poly}(\eps)g)$, where poly$(\eps)$ is a polynomial in $\eps$; this indicates that $f=O(g)$ for fixed $\eps$, which is not necessarily the case if $\eps$ is allowed to vary.

{\color{black}We assume familiarity with the complexity classes P and NP. The complexity class DTIME$(f(n))$ (resp. NTIME$(f(n))$) represents the set of decision problems decidable by a deterministic (resp. non-deterministic) Turing machine in $O(f(n))$ time. The complexity class BPTIME($f(n)$) represents the set of decision problems $L$ solvable by a probabilistic algorithm in $O(f(n))$ time, such that for all $x \in L$, the algorithm accepts $x$ with probability at least $\frac{2}{3}$, and for all $x \not\in L$, the algorithm accepts $x$ with probability at most $\frac{1}{3}$.

Given an NP--hard optimization problem $P$, let $\mathcal{I}$ be the set of instances of $P$. If $P$ is a minimization problem and $\alpha \ge 1$, we say that a polynomial-time algorithm $A$ is an $\alpha$--approximation if, for all instances $I \in \mathcal{I}$, $A(I)$ returns a feasible solution with cost $c$ such that $\OPT_I \le c \le \alpha \OPT_I$ where $\OPT_I$ denotes the cost of the optimum solution for instance $I$. If $P$ is a maximization problem, then $A$ is an $\alpha$--approximation ($\alpha \le 1$) if $\alpha \OPT_I \le c \le \OPT_I$ for all instances $I$. Note that $\alpha$ may be a function of the size of the instance, or some other parameter related to the instance (such as the maximum vertex degree $\Delta$ of the input graph). In rarer cases, the approximation ratio given may be additive instead of multiplicative.}


\section{Graph Spanners}\label{SEC:Spanners}
Given a graph $G$, possibly edge-weighted, a \emph{graph spanner} (or \emph{spanner} for short) is a subgraph $G'$ which preserves lengths of shortest paths in $G$ up to some error or distortion, e.g., additive and/or multiplicative error.  There are several definitions and formulations of graph spanners.  Most spanners can be specified using two parameters: a \textit{distortion function} $f$ which characterizes how much distances from the original graph are allowed to be distorted, and a subset $P\subseteq V\times V$ which prescribes pairs of vertices of the initial graph for which the distances need to be approximately preserved.  That is, a subgraph $G'=(V',E')$ is a spanner with distortion $f$ for $P\subseteq V\times V$ if
\[ d_{G'}(u,v)\leq f\left(d_G(u,v)\right)\]
for all $(u,v) \in P$. It is worth noting that $d_G(u,v)\leq d_{G'}(u,v)$ for any subgraph $G'$ of $G$; as such, we stipulate that the distortion function satisfies $f(x) \ge x$.

The following are typical choices for the distortion function $f$ {\color{black}\st{(from most to least general)}}.
\begin{itemize}\normalsize
\item \textbf{Multiplicative} (or \textbf{$t$-spanner}): distortion given by $f(x)=tx$ for some constant $t\geq1$, where $t$ is called the stretch factor.  That is, distances are \textit{stretched} by no more than a factor of $t$. Some authors use $k$ instead of $t$ to denote the stretch factor.
\item \textbf{Additive}: distortion given by $f(x) = x+\beta$.  In other words, distances in additive graph spanners are not \textit{elongated} more than $\beta$ units (or edges). These subgraphs preserve long distances with ratio close to 1. Additive spanners are sometimes called $+\beta$--spanners.
\item \textbf{Linear} (or \textbf{$(\alpha,\beta)$-spanner}): distortion given by $f(x) = \alpha x+\beta$ for $\alpha,\beta\geq0$. {\color{black}In particular, a $t$-spanner is a $(t,0)$-spanner, and an additive $+\beta$-spanner is a $(1,\beta)$-spanner.}
\item \textbf{Sublinear:} distortion given by $f(x) = x + o(x)$.  In particular, distortions of the form $f(x)=x+O(x^{1-\frac1k})$ for positive integers $k$ are of particular interest.
Roughly, this is because this distortion function arises as an \emph{information-theoretic} barrier to compressing graph distances; that is, the most space-efficient data structures that approximate graph distances have essentially these error functions (but it is still open whether they are right for \emph{spanners}; see \cite{abboud2018hierarchy} or Section \ref{SEC:emulators} for further discussion).
Like linear error, sublinear error can be viewed as a compromise between multiplicative and additive distortion, as it stretches small distances up to a constant multiplicative factor while the stretch factor for long distances approaches 1.
\item \textbf{Distance Preservers:}  no distortion, i.e., $f(x)=x$.  These are $(1,0)$--spanners.
\end{itemize}
Next, we illustrate different terminologies based on the subset $P\subseteq V\times V$ chosen (from most to least general).
\begin{itemize}\normalsize
\item \textbf{Pairwise Spanners:} the given spanner condition {\color{black}must hold for specific pairs of vertices $P\subseteq V\times V$} (note that $P$ is not necessarily symmetric).  
\item \textbf{Sourcewise Spanners:} the given spanner condition must hold for $P=S\times V$ for some subset $S\subseteq V$.
\item \textbf{S--T (Source--Target) Spanners:}  the given spanner condition must hold for $P=S\times T$ for subsets $S,T\subseteq V$ (not necessarily disjoint).
\item \textbf{Subsetwise Spanners:} the given spanner condition must hold between all vertices in a fixed subset $S\subseteq V$.  That is, subsetwise spanners are pairwise spanners where $P=S\times S\subseteq V\times V$.
\item \textbf{Spanners:} with no additional terminology, it is typically implied that $P=V\times V$, i.e., distances are approximately preserved for all pairs of vertices.
\end{itemize}

Most spanner problems can be specified via the distortion parameters ($\alpha,\beta$) and the pairs $P$ of vertices for which distances must be preserved. To make notation concise, we typically write $(\alpha,\beta,P)$--spanner, or $(f,P)$--spanner for general distortions.

Note also that while all variants of spanners are well-defined for weighted and unweighted graphs, it is somewhat more natural to consider additive spanners when the graph is unweighted.  Indeed, if a graph $G$ is weighted with arbitrarily large edge weights, then there exists $\beta$ such that there is no additive $\beta$--spanner of $G$ {\color{black} except for $G$ itself.}

\subsection{Classical Results} \label{SEC:classical}

\begin{itemize}
\setlength{\itemindent}{.2in}
\item[\cite{peleg1989graph}] \bibentry{peleg1989graph}
\item[\cite{althofer1993sparse}] \bibentry{althofer1993sparse}
\item[\cite{Kortsarz94Algorithms}] \bibentry{Kortsarz94Algorithms}
\end{itemize}

The classical and most common spanner problem is {\color{black} that of computing a sparse multiplicative $t$--spanner of an input graph, often called ``basic $t$-- (or $k$--) spanner'' extensively throughout the literature.} It is usually defined as follows:
\begin{problem}[Basic $t$--Spanner Problem]
Given a connected graph $G=(V,E)$ and a fixed $t\geq1$, find a subset $E'\subseteq E$ such that the subgraph $G'=(V,E')$ of $G$ satisfies
\begin{equation}\label{EQ:tspanner} d_{G'}(u,v)\leq t\cdot d_{G}(u,v),\quad \text{for all } u,v\in V.
\end{equation}
\end{problem}
The notion of a $t$--spanner was introduced by Peleg and Sch\"{a}ffer in \cite{peleg1989graph} (see also Peleg and Ullman \cite{peleg1989optimal}), though the idea also appeared implicitly in earlier work of Awerbuch \cite{awerbuch1985complexity} and Chew \cite{Chew89}.
Peleg and Sch\"{a}ffer \cite{peleg1989graph} show that for unweighted graphs, determining if a $t$--spanner of $G$ containing at most $m$ edges exists is NP--complete.  They also discuss at length a reduction of the $t$--spanner problem to particular classes of graphs -- chordal graphs -- and show that the generic lower bounds for the number of edges required to form a $t$--spanner for an arbitrary graph may be significantly improved for restricted graph classes. 

Note that in order to check whether $G'$ is a $t$--spanner of $G$, one only needs to check that inequality \eqref{EQ:tspanner} holds for all edge pairs $(u,v)$ where $uv \in E$, rather than for all $\binom{|V|}{2}$ vertex pairs, which is stated in the following simple proposition:

\begin{proposition}[{\cite[Lemma 2.1]{peleg1989graph}}]\label{PROP:tspanner verification}
$G'$ is a $t$--spanner of $G$ if and only if $d_{G'}(u,v)\leq t\cdot d_G(u,v)$ for all $uv\in E$.
\end{proposition}

\begin{proof}
The forward direction is obvious. For the reverse direction, let $u$ and $v$ be distinct vertices in $V$, and let $u_0 u_1 \ldots u_m$ be a shortest $u$-$v$ path in $G$, where $u_0 = u$ and $u_m = v$.  Then note that by assumption,
\[d_{G'}(u,v) \leq \sum_{i=0}^{m-1} d_{G'}(u_i,u_{i+1}) \leq t\cdot\sum_{i=0}^{m-1} d_G(u_i,u_{i+1}) = t\cdot d_G(u,v).\]
\end{proof}

Thus, one can check the spanner inequality \eqref{EQ:tspanner} for  $|E|$ pairs of vertices rather than for all $\binom{|V|}{2}$ vertex pairs to verify that a given subgraph is a $t$--spanner.

\subsection{Sparsity and Lightness}

Alth\"{o}fer et al. \cite{althofer1993sparse} discuss a specific aspect of the $t$--spanner problem, namely \textit{sparsity}.  The idea is that a good $t$--spanner of a graph should contain very few edges while still approximately preserving distances in $G$. {\color{black}The sparsity of a spanner $G'=(V,E')$ is often defined as} \[ \text{sparsity}(G') = \frac{|E'|}{|E|}. \] 
{\color{black}Thus $0\leq\textnormal{sparsity}(G')\leq 1$ for any subgraph $G'$ of $G$.} Alth\"{o}fer et al. \cite{althofer1993sparse} prove many lower and upper bounds for general graphs on the sparsity of a $t$--spanner, and also give a simple polynomial time algorithm for producing such a sparse spanner.  One advantage of their method is that the algorithm presented provides a $t$--spanner not only with few edges, but also whose weight is comparable to the weight of the minimum spanning tree (in the case that $G$ is a weighted graph).  We will further discuss their results in Section \ref{SEC:Greedy}. 

For weighted graphs, a more natural consideration is the \textit{lightness} of a spanner, which is related to the total weight of the spanner. {\color{black} It is usually compared with the weight of the minimum spanning tree.} We may thus define {\color{black}\[ \textnormal{lightness}(G') = \frac{W(E')}{W(\MST(G))}.\]
Note that $\textnormal{lightness}(G')\geq1$ since the MST is the sparsest subgraph connecting all the nodes.}

Often we are interested in finding the sparsest or lightest 
spanner of a given graph:

\begin{problem}[Sparsest/Lightest Spanner problem]
Given a graph $G=(V,E)$, distortion $f$, and $P\subseteq V\times V$, find $G'=(V',E')$ such that 
\begin{enumerate}
    \item $G'$ is a $(f,P)$--spanner for $G$, and
    \item $\sparsity(G')\leq\sparsity(H)$ for all other $(f,P)$--spanners $H$ if $G$ is unweighted, or
    \item $\lightness(G')\leq\lightness(H)$ for all other $(f,P)$--spanners $H$ if $G$ is weighted.
\end{enumerate}
\end{problem}

Kortsarz and Peleg~\cite{Kortsarz94Algorithms} show that the problem of finding the sparsest 2--spanner of an unweighted graph $G=(V,E)$ admits a polynomial time $\log\left(\frac{|E|}{|V|}\right)$--approximation. Note that this problem is equivalent to finding an edge set $E' \subseteq E$ such that for every edge $e \in E \setminus E'$, there is a triangle (3--cycle) in $G$ containing $e$ whose remaining two edges belong to $E'$.

To describe their approximation algorithm, note that given $U \subseteq V$, the \emph{density} of $U$, denoted $\rho_G(U)$, is defined by $\rho_G(U) = \frac{|E(U)|}{|U|}$, where $(U,E(U))$ is the subgraph of $G$ induced by $U$. The \emph{maximum density} problem is to find vertex subset $U \subseteq V$ such that $\rho_G(U)$ is maximized, which can be solved in polynomial time.

The algorithm given in~\cite{Kortsarz94Algorithms} maintains three sets of edges: $H^s$, the set of edges in the spanner; $H^c$, the set of ``covered edges'' (edges either in the spanner, or edges on a triangle containing two edges in $H^s$); $H^u$, the remaining unspanned edges. Given a set $H^u$ of unspanned edges and vertex $v$, the authors denote by $N(H^u, v)$ the subgraph of $G$ whose vertex set is $N(v)$\footnote{the set of neighbors of $v$ in $G$}, and whose edge set is the set of edges in $G$ induced by $N(v)$, which are also unspanned (i.e., $E(N(v)) \cap H^u$). The authors denote by $\rho(H^u, v)$ the maximum density of this neighborhood graph over all vertices $v \in G$.

The idea of the algorithm is: while $\rho(H^u, v) > 1$, find a vertex $v \in G$ such that $\rho(H^u, v)$ is maximum. Then within this restricted neighborhood graph, solve the maximum density problem to find a corresponding dense subset $U_v$ of neighbors of $v$. Add the edges from $v$ to each $u \in U(v)$ to the set of spanner edges $H^s$. The edges ``covered'' by the newly-added edges in $H^s$ are added to $H^c$, and $H^u$ also updates accordingly.

Since at least one edge is added to $H^c$ at each iteration, the number of iterations is at most $m$. The authors show that, using the maximum density problem as a subroutine, a sparse 2-spanner with approximation ratio $O\left(\log\frac{|E|}{|V|}\right) = O(\log n)$ can be found in $O(m^2n^2 \log(\frac{n^2}{m}))$ time.

\subsection{Relation to Minimum Spanning Trees}
Recall that $G'$ is a \emph{spanning tree} of $G$ if $G'$ is connected, acyclic, and spans all vertices of $G$. The well-known \emph{minimum spanning tree} (MST) problem is to compute a spanning tree of minimum total weight. This can easily be done in polynomial time (e.g., using Kruskal's or Prim's algorithm). We denote by $W(\MST(G))$ the weight of an MST of $G$.


Note that the MST of a graph can be an arbitrarily poor spanner; for example, {\color{black} let $G$ be an unweighted cycle with $|V|$ edges.  The MST is obtained by deleting any edge; however, the MST is only a $(|V|-1)$--spanner of the original graph.} 
The MST problem can be interpreted as a special case of the $t$--spanner problem where $t$ is arbitrarily large. 

\subsection{Erd\H{o}s' Girth Conjecture}
\label{SEC:girth} 
\begin{itemize}
\setlength{\itemindent}{.2in}
\item[\cite{ErdosGirth}]{\bibentry{ErdosGirth}}
\item[\cite{ErdosGirth2}]{\bibentry{ErdosGirth2}}
\item[\cite{Parter14}]{\bibentry{Parter14}}
\end{itemize}

Recall that the {\em girth} of a graph $G$ is defined as the length of the smallest cycle in $G$, and is infinity if $G$ is acyclic.
Letting $\gamma(n, k)$ denote the maximum possible number of edges in an $n$--vertex graph with girth $>k$, the authors in \cite{althofer1993sparse} give a simple algorithm that, for any $n$--vertex (undirected, possibly weighted) input graph and positive integer $t$, produces a $t$--spanner on $|E'| \le \gamma(n, t+1)$ edges.
This is best possible in the following sense:
\begin{proposition}\label{PROP:Girth}
An undirected unweighted graph $G = (V, E)$ of girth $>t+1$ has no proper subgraph that is a $t$--spanner.
\end{proposition}
\begin{proof}
Consider any edge $uv \in E$.
If $uv$ is removed from $G$, then $\dist_G(u, v)$ changes from $1$ to one less than the length of the shortest cycle containing $uv$ in the original graph.
Hence, if $G$ has girth greater than $t+1$, then $\dist_{G \setminus \{uv\}}(u, v) > t$, so the spanner property is not satisfied for the pair {\color{black}$(u,v)$} after the edge is removed.
\end{proof}

Thus, if we consider an $n$--vertex graph $G = (V, E)$ with girth greater than $t+1$ and $|E| = \gamma(n, t+1)$ edges, the only $t$--spanner of $G$ is $G$ itself.
It follows that no algorithm can improve in general on the density bound of \cite{althofer1993sparse} with, say, a bound of $|E'| \le \gamma(n, t+1) -1$.

It still remains a major open problem to determine $\gamma(n, k)$, even asymptotically, though upper bounds called the \emph{Moore Bounds} are given by a folklore counting argument:
\begin{proposition}\label{PROP:Moore}
$\gamma(n,k) = O\left(n^{1 + \frac{1}{\lfloor k/2 \rfloor}}\right).$
\end{proposition}
\begin{proof}[Sketch of Proof]
See \cite{althofer1993sparse} for full detail.
Let $G = (V, E)$ be a graph with $|V|=n$ and girth $>k$, and let $d$ be its average vertex degree.
There are two steps in the proof.
First, we find a nonempty subgraph $H \subseteq G$ of minimum degree $\Omega(d)$.
We generate $H$ by iteratively deleting any vertex in $G$ of degree $\le \frac{d}{4}$; one counts that at most $\frac{|E|}{2}$ edges are deleted in total, so $H$ remains nonempty.
Second, pick an arbitrary vertex $v$ in $H$ and consider a BFS tree $T$ rooted at $v$ to depth $\lfloor \frac{k}{2} \rfloor$.
Since $H$ has girth $>k$, each vertex in $T$ can have only one edge to a vertex in the same or previous layer of the tree (else we would have a cycle of length $\le k$).
Since every vertex in $T$ has degree $\Omega(d)$, one counts that there are $\Omega(d)^{\lfloor \frac{k}{2} \rfloor}$ total vertex in $T$.
Thus $d = O\left(n^{\frac{1}{\lfloor k/2 \rfloor}}\right)$, so $|E| = O\left(n^{1 + \frac{1}{\lfloor k/2 \rfloor}}\right)$.
\end{proof}

Erd\H{o}s' Girth Conjecture is the statement that the Moore Bounds are tight; that is, $\gamma(n, k) = \Omega(n^{1 + \frac{1}{\lfloor k/2 \rfloor}})$.\footnote{The conjecture comes from \cite{ErdosGirth}, page 5, where Erd\H{o}s wrote ``It seems likely that'' a certain equation holds (7) which is equivalent to the statement $\gamma(n, k) = \Omega_k(n^{1 + \frac{1}{\lfloor k/2 \rfloor}})$.  Most modern papers quote a strengthened version of this statement where the implicit constant may not even depend on $k$.}
The conjecture generally plays two important roles in the literature on spanners.
First, in some applications of spanners it is useful to know the number of edges the spanner could possibly have, and this currently requires the Girth Conjecture.
Second, there are some known constructions of $t$--spanners on $O(n^{1 + \frac{1}{\lfloor (t+1)/2 \rfloor}})$ edges with certain desirable properties relative to \cite{althofer1993sparse} (for example, Baswana and Sen \cite{Surender03ALP} give a linear time algorithm to produce spanners of this quality).
Whereas \cite{althofer1993sparse} produces spanners with optimal size/distortion tradeoff \emph{regardless} of the truth of the Girth Conjecture, these other constructions have optimal tradeoff only if the Girth Conjecture is assumed.
Thus, many natural algorithmic questions about spanners (such as linear time computation of spanners with optimal size/distortion tradeoff) are currently closed only if the Girth Conjecture is assumed.


So far, we have discussed the Girth Conjecture with respect to multiplicative spanners, but it more generally constrains spanners with additive or mixed error in the same way.
Specifically, assuming the Girth Conjecture, any construction of $(\alpha,\beta)$--spanners of size $O(n^{1+\frac1k})$ must have $\alpha+\beta\geq 2k-1$. The proof of this is essentially the same as that of Proposition \ref{PROP:Girth}.
Independent of the veracity of the girth conjecture, Woodruff~\cite{woodruff2006lower} proved this fact in the setting $\alpha=1$: for any $k\in\N$, there exists a graph on $n$ nodes for which any $(1,2k-2)$--spanner has $\Omega(k^{-1}n^{1+\frac1k})$ edges.

\subsection{Open Problems}
\begin{enumerate}

\item Considering discontinuous distortion functions would also be of interest, e.g., one could allow small distances to be distorted more than large ones, or require distance preservation of close vertices, but small distortion for far away ones.

\item Erd\H{o}s' Girth Conjecture is known to be true for $k=1, 2, 3$, or $5$ \cite{Wenger91, Tits59}.  Is it true for other values of $k$?
\end{enumerate}


\section{NP-hardness and Hardness of Approximation}\label{SEC:Complexity}
{\color{black}Most interesting spanner problems are NP-complete. In this section, we focus primarily on hardness results for the basic $t$--spanner and additive spanner problems. For hardness results on other types of spanner problems (e.g., tree $t$--spanner or lowest-diameter $t$--spanner), see Section~\ref{SEC:TypesOfSpanners}.}

\subsection{NP-hardness of the basic $t$--spanner problem}

\begin{itemize}
\setlength{\itemindent}{.2in}
\item[\cite{peleg1989graph}] \bibentry{peleg1989graph}
\item[\cite{cai1994np}] \bibentry{cai1994np}
\end{itemize}

 Peleg and Sch\"{a}ffer~\cite{peleg1989graph} show that, given an unweighted graph $G$, and integers $t,m' \ge 1$, determining if $G$ has a $t$--spanner containing $m'$ or fewer edges is NP--complete, even when $t$ is fixed to be 2. The reduction is from the edge dominating set\footnote{An \emph{edge dominating set} of a graph $H=(V,E)$ is a subset of edges $E' \subseteq E$ such that every edge $e \in E\setminus E'$ is adjacent to at least one edge in $E'$.} (EDS) problem on bipartite graphs, defined as follows: given a bipartite graph $H=(V,E)$ with bipartition $(X,Y)$, and an integer $K \ge 1$, determine if $H$ has an edge dominating set consisting of at most $K$ edges. The reduction can be summarized as follows: given an instance $\langle H=(V,E), K \rangle$ to EDS where $H$ has bipartition $(X,Y)$, $X = \{x_1, \ldots, x_{n_x}\}$, and $Y = \{y_1, \ldots, y_{n_y}\}$, let $G(H)$ be constructed by adding a vertex $x_{ij}$ for each $1 \le i < j \le n_x$, and similarly for $Y$. Let $V(G(H)) = X \cup Y \cup \{x_{ij} \mid 1 \le i < j \le n_x \} \cup \{y_{ij} \mid 1 \le i < j \le n_y\}$.

For each $1 \le i < j \le n_x$, add edges $x_ix_j$, $x_ix_{ij}$, and $x_jx_{ij}$, and do so similarly for $Y$. The authors define $E_{XY}$ as the union of all of these added edges. Let $E(G(H)) = E \cup E_{XY}$, $t=2$, and $m' = K + n_x(n_x - 1) + n_y(n_y - 1)$. One can show that $\langle H, K \rangle$ is accepted iff $\langle G(H), t, m' \rangle$ is accepted.

Cai~\cite{cai1994np} also shows that for any fixed $t \ge 2$, the minimum $t$--spanner problem is NP--hard, and for $t\ge 3$, the problem is NP--hard even when the input is restricted to bipartite graphs.  The reduction is from the 3--SAT problem.

\subsection{Label Cover and hardness of approximation of the basic $t$--spanner problem}
\begin{itemize}
\setlength{\itemindent}{.2in}
\item[\cite{arora96hardness}] \bibentry{arora96hardness}
\item[\cite{elkin2005approximating}] \bibentry{elkin2005approximating}
\item[\cite{Elkin2007}] \bibentry{Elkin2007}
\item[\cite{dinitz2016label}] \bibentry{dinitz2016label}
\end{itemize}
For $k=2$, the basic $k$--spanner problem admits an $O(\log n)$ approximation~(\cite{Kortsarz94Algorithms}, Section~\ref{SEC:classical}). {\color{black}Dinitz et al.~\cite{dinitz2016label} show that for $k \ge 3$, and for all $\varepsilon >0$, the basic $k$--spanner problem cannot be approximated with ratio better than $2^{(\log^{1-\varepsilon}n)/k}$ unless $\text{NP} \subseteq \text{BPTIME}(2^{\text{polylog}(n)})$ (This result was initially claimed by Elkin and Peleg~\cite{elkin2000strong}, but retracted due to error).} Elkin and Peleg~\cite{elkin2005approximating} propose approximation algorithms with sublinear approximation ratio, and study certain classes of graphs for which logarithmic approximation is feasible.  Elkin and Peleg~\cite{Elkin2007} extend the hardness results to other spanner problems, and also show strong inapproximability for the \emph{directed} unweighted $k$--spanner problem.

{\color{black}Many strong hardness results rely on the hardness of \emph{Label Cover}, initially defined by Arora and Lund~\cite{arora96hardness}, which may be stated as follows.

\begin{definition}[Label Cover]\label{def:labelcover}
Let $G=(V,E)$ be a biregular bipartite\footnote{A \emph{biregular bipartite} graph is a graph with bipartition $V=V_1 \cup V_2$ such that every vertex in $V_1$ has the same degree, and every vertex in $V_2$ has the same degree.} graph with bipartition $V=V_1 \cup V_2$. Additionally, let $\Sigma_1$ and $\Sigma_2$ be two alphabets representing the sets of possible labels associated with $V_1$ and $V_2$. For every edge $e \in G$, there is a nonempty relation $\pi_e \subseteq \Sigma_1 \times \Sigma_2$. A \emph{labeling} is an assignment of one or more labels to every vertex in $V$. We say that edge $e = uv$ ($u \in V_1$, $v \in V_2$) is \emph{covered} by the labeling if there is a label $a_1$ assigned to $u$ and a label $a_2$ assigned to $v$, such that $(a_1, a_2) \in \pi_e$.
\end{definition}
Often, Label Cover is stated as a maximization problem: every vertex in $G$ is assigned exactly one label, and the goal is to maximize the number of covered edges. It is known that both the maximization and minimization versions are quasi-NP--hard to approximate with ratio $2^{\log^{1-\varepsilon} n}$ for all $\varepsilon > 0$~\cite{arora96hardness}.}

Dinitz et al.~\cite{dinitz2016label} show that if there is an additional requirement that the girth of $G$ is larger than $k$, where $3 \le k \le \log^{1-2\varepsilon}n$, then the maximization version of Label Cover cannot be approximated with ratio better than $2^{(\log^{1-\varepsilon}n)/k}$ unless $\text{NP} \subseteq \text{BPTIME}(2^{\text{polylog}(n)})$, for sufficiently large $n$. By considering a variant of the minimization problem of Label Cover (Min-Rep), they show the hardness of the basic $k$--spanner problem; specifically, it is also hard to approximate the basic $k$--spanner problem within a factor better than $2^{(\log^{1-\epsilon}n)/k}$. 

\subsection{NP--Hardness for Planar Graphs}
\begin{itemize}
\setlength{\itemindent}{.2in}
    \item[\cite{Brandes1997NpcompletenessRF}] \bibentry{Brandes1997NpcompletenessRF}
    \item[\cite{kobayashi2018np}] \bibentry{kobayashi2018np}
\end{itemize}
Brandes and Handke~\cite{Brandes1997NpcompletenessRF} show that the basic $t$--spanner problem is NP--complete for fixed $t \ge 5$ when the input graph is planar and unweighted, and is also NP--complete for fixed $t \ge 3$ when the input graph is planar and weighted. For unweighted graphs, the reduction is from \emph{planar 3--SAT}, a variant of 3--SAT where the underlying bipartite graph induced by the variables and clauses is planar. Kobayashi~\cite{kobayashi2018np} recently showed that the minimum $t$--spanner problem is NP--hard on planar graphs for $t \in \{2,3,4\}$. 

\subsection{NP--Hardness of Additive Spanners}
\begin{itemize}
\setlength{\itemindent}{.2in}
    \item[\cite{Brandes1997NpcompletenessRF}] \bibentry{liestmannp}
    \item[\cite{kobayashi2019fpt}] \bibentry{kobayashi2019fpt}
\end{itemize}
Liestman and Shermer~\cite{liestmannp} show that for all integers $\beta \ge 1$, determining if a graph $G$ contains an additive $+\beta$-spanner containing $m'$ or fewer edges is NP--hard via a reduction from the edge dominating set problem.

Kobayashi~\cite{kobayashi2019fpt} considers a parametrized version of the additive $(1,\beta)$--spanner problem where the number of removed edges is regarded as a parameter $k$, and a fixed-parameter algorithm is given for it. The main result is that there exists a fixed-parameter tractable algorithm for the Parameterized Minimum Additive $(1,\beta)$--spanner problem that runs in $(\beta+ 1)^{O(k^2+\beta k)}mn$, or $2^{O(k^2)}mn$ if $\beta$ is fixed.
These results are generalized for $(\alpha, \beta)$--spanners.
{\color{black} However, we remark that these algorithms run in polynomial time only if one wishes to remove only a constant number of edges.
In many settings of interest one hopes for a spanner that is much sparser than the input graph, and for these a different algorithmic paradigm is needed.}

\subsection{Open Problems}
\begin{enumerate}
    \item Kobayashi~\cite{kobayashi2018np} leaves as an open question whether the minimum $t$--spanner problem on bounded-degree graphs of degree at most $\Delta$ is NP--hard for certain fixed $t$ and $\Delta$, namely $(t,\Delta) = (2,5)$, $(2,6)$, $(2,7)$, $(3,4)$, $(3,5)$, $(4,3)$, $(4,4)$, $(4,5)$.
\end{enumerate}

\section{The Greedy Algorithm for Multiplicative Spanners}\label{SEC:Greedy}

\begin{itemize}
\setlength{\itemindent}{.2in}
    \item[\cite{althofer1993sparse}] \bibentry{althofer1993sparse}
    \item[\cite{filtser2016greedy}] \bibentry{filtser2016greedy}
\end{itemize}

One of the original methods for constructing spanners is to use a greedy algorithm.  The essential idea is to first sort the edges in nondecreasing order by weight, choose the edge with the smallest weight first, and then each subsequent edge is chosen or not according to some criteria which guarantees that the end result is a spanner of the desired type.  Despite being the oldest construction of spanners, greedy algorithms remain one of the most utilized methods for achieving this task.  

\subsection{Kruskal's Algorithm for Computing Minimum Spanning Trees}

There are many algorithms for computing MSTs, including Bor\r{u}vka's, Prim's, and Kruskal's algorithm.  All of these are examples of greedy algorithms.  Since it is the basis for the greedy algorithm to construct $t$--spanners, we review Kruskal's algorithm in Algorithm~\ref{ALG:Kruskal}. 

\begin{algorithm}[h!]
\caption{\textsc{MST}($G = (V, E)$); Kruskal's MST Algorithm}\label{ALG:Kruskal}
\begin{algorithmic}
\State Sort edges in nondecreasing order of weight
\State $G' = (V, E' \gets \emptyset)$
\For{$uv \in E$}
\If{$\dist_{G'}(u,v) = \infty$ (i.e. $u, v$ are previously disconnected in $G'$)}

\State $E'\gets E'\cup\{uv\}$
\EndIf   
\EndFor
\State \textbf{return} $G'$
\end{algorithmic}
\end{algorithm}

\subsection{\label{SEC:greedymult}The Greedy Algorithm for Multiplicative Spanners}

Alth\"{o}fer et al. \cite{althofer1993sparse} proposed and analyzed the first greedy algorithm for computing a sparse $t$--spanner of a weighted graph; see Algorithm~\ref{ALG:Greedy t spanner}.

\begin{algorithm}[h!]
\caption{\textsc{GreedySpanner}($G = (V, E)$, $t$); Greedy $t$--spanner Algorithm }\label{ALG:Greedy t spanner}
\begin{algorithmic}
\State Sort edges in nondecreasing order of weight
\State $G' = (V, E' \gets\emptyset)$
\For{$uv\in E$}
\If{$\dist_{G'}(u,v) > t\cdot w_{uv}$}
\State $E'\gets E'\cup\{uv\}$
\EndIf
\EndFor
\State \textbf{return} $G'$
\end{algorithmic}
\end{algorithm}

Note first, that Algorithm \ref{ALG:Greedy t spanner} guarantees that the lowest weighted edge is chosen.  Likewise, if at any iteration in the for loop there is no path in $G'$ from $u$ to $v$, then the edge $uv$ is chosen.  This algorithm is based on Kruskal's algorithm, but the key difference is that Algorithm \ref{ALG:Greedy t spanner} does not enforce the condition that $G'$ is acyclic.  Indeed, suppose edges $u_1u_2$ and $u_1u_3$ have been chosen already to be in $E'$, and that $u_2u_3$ is an edge in the original graph $G$, and the shortest path in $G$ from $u_2$ to $u_3$ is the two-edge path going from $u_2$ to $u_1$ to $u_3$.  Kruskal's algorithm would reject the new edge because its endpoints are already connected.  However, provided $t\cdot w_{u_2u_3}<w_{u_2u_1}+w_{u_1u_3}$, this edge would be added to $E'$.

That said, an analogous cycle-free property holds for the greedy algorithm: a graph returned by Algorithm \ref{ALG:Greedy t spanner} with parameter $t$ will not have any cycles on $\le t+1$ edges.
To see this, let $C$ be a cycle on $\le t+1$ edges in the input graph $G$, and let $uv \in C$ be the last edge considered by the greedy algorithm.
When edge $uv$ is considered, either another edge in $C$ has been discarded, or else (due to the edge ordering) there is a $u$-$v$ path through $C$ of length $\le t \cdot w_{uv}$, and thus we will discard $uv$.
In either case, $C$ does not survive in the output graph $G'$.
Thus, one can reasonably view Kruskal's algorithm as the special case of Algorithm \ref{ALG:Greedy t spanner} with $t = \infty$.

Alth\"{o}fer et al.~prove that any $G'$ constructed from Algorithm \ref{ALG:Greedy t spanner} is a $t$--spanner for $G$.  However, we note here that their proof is actually not dependent upon the greedy reordering of the edges, which is an interesting fact in its own right.

\begin{proposition}[\cite{althofer1993sparse}]\label{PROP:Greedy Ordering}
Algorithm \ref{ALG:Greedy t spanner} yields a $t$--spanner for $G$ regardless of the ordering of the edges in the first step.  
\end{proposition}
\begin{proof}
First, we may assume without loss of generality that for each edge $uv$ in the input graph $G$, we have $\dist_G(u, v) = w_{uv}$.
Otherwise, we may remove $uv$ from $G$ without changing its shortest path metric at all, and thus any spanner of the remaining graph is also a spanner of $G$ itself.

For each $uv \in E$, when we consider $uv$ in the greedy algorithm, we either have
$$\dist_{G'}(u, v) \le t \cdot w_{uv} = t \cdot \dist_G(u, v),$$
or else we add $uv$ to $G'$ and thus have $\dist_{G'}(u, v) = \dist_G(u, v)$.
In either case, the pair $(u,v)$ satisfies the spanner property.
The proposition then follows from Proposition \ref{PROP:tspanner verification}.
\end{proof}

Despite the fact that Algorithm \ref{ALG:Greedy t spanner} yields a $t$--spanner for any ordering of the edges, the rest of the analysis of Alth\"{o}fer et al. crucially hinges upon the greedy edge ordering.  In \cite[Lemma 2]{althofer1993sparse}, it is shown that the output $G'$ of Algorithm \ref{ALG:Greedy t spanner} is such that its girth is at least $t+1$.  However, this need not hold if we allow the algorithm to run with a different edge ordering as the example in Figure \ref{FIG:AltGirth} demonstrates. 

\begin{figure}[!h]
\centering
\includegraphics[width=\textwidth]{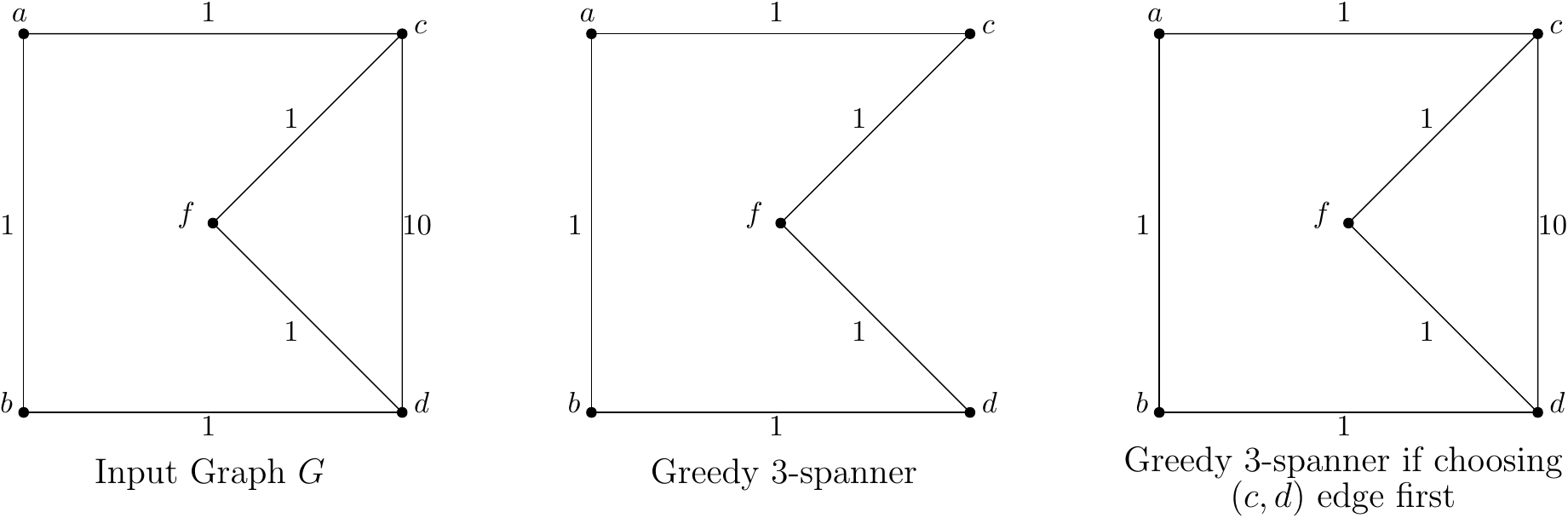}
\caption{Ordering the edges by their weights is pivotal for the girth property of the greedy $t-$spanner. The rightmost graph would be output of the algorithm for $t=3$ if we allowed non-ordered edges by picking the edge $cd$ with weight 10. This graph has girth 3, which is smaller than $(t+1)=4$.\label{FIG:AltGirth}}
\end{figure}

\begin{remark}\label{REM:Greedy}
It should be noted that the greedy algorithm presented here essentially relies on Proposition \ref{PROP:tspanner verification}, which allows one to only enforce the $t$--spanner condition on edges in the initial graph $G$.  This proposition does not hold for additive spanners (for example), and so to find a greedy algorithm to produce an additive (or mixed) spanner, something else would need to be done.
Indeed, \cite{Knudsen14} might be viewed as a greedy algorithm for additive spanners, but for this reason it requires an extra preprocessing step before the greedy part, and thus is not a direct analog of the multiplicative greedy algorithm discussed here.
\end{remark}

\subsection{Existential Optimality of Greedy $t$--spanners}

Garay et al.~\cite{garay1998sublinear} (see also \cite[Chapter 24]{PelegBook}) distinguish between different notions of optimality of a given algorithmic output.  The essential difference is one of a single quantifier: namely one is a ``for all'' statement and the other is a ``there exists'' statement.  An algorithmic solution is {\em universally optimal} for a given class of graphs $\mathcal{G}$ if for every graph $G\in\mathcal{G}$, the algorithm gives the optimal solution.  On the other hand, an algorithm is {\em existentially optimal} for a class of graphs $\mathcal{G}$ provided there exists a graph $G\in\mathcal{G}$ for which the algorithm constructs the optimal solution.  

Filtser and Solomon \cite{filtser2016greedy}, following \cite{althofer1993sparse}, prove existential optimality of the greedy $t$--spanner; although they prove something somewhat stronger than the notion of Garay et al. described above.  They prove that the greedy $t$--spanner for a graph $G$ in a given (fixed) class of graphs $\mathcal{G}$ is never worse than the worst-case optimal solution for the whole class $\mathcal{G}$.  For example, given a class of graphs $\mathcal{G}$ on $n$ vertices, the worst-case number of edges in an optimal $t$--spanner is given by
\[ m(t):=\underset{G\in\mathcal{G}}\sup\;\underset{G'\text{ is a }t\text{-spanner for }G}{\underset{G'=(V,E'),}{\inf}}\;|E'|.\]  Then the greedy $t$--spanner algorithm is such that for every $G\in\mathcal{G}$, the algorithm's output is a $t$--spanner with at most $m(t)$ edges.  In particular, this implies that the greedy $t$--spanner is existentially optimal because there exists a graph $G$ whose optimal $t$--spanner has $m(t)$ edges, and its greedy $t$--spanner has at most $m(t)$ edges as well, and hence is an optimal solution.

In our setting, optimality of a $t$--spanner construction will be considered via two parameters, the number of edges in the spanner, and its total weight.  Filtser and Solomon prove that the greedy $t$--spanner algorithm is existentially optimal.  We state their result in a stronger way here, but in such a way that the proof becomes much simpler than the original.  To begin, we state a crucial lemma which is interesting in its own right, which states that the only $t$--spanner of a greedy $t$--spanner is itself.

\begin{lemma}[{\cite[Lemma 3]{filtser2016greedy}}]\label{LEM:Greedy}
Let $t\geq1$ be fixed, and let $G$ be any weighted connected graph with $n$ vertices.  Let $G'$ be the greedy $t$--spanner of $G$.  If $G''$ is a $t$--spanner of $G'$, then $G''=G'$.
\end{lemma}
\begin{proof}
Suppose $G'=(V,E')$ and $G''=(V,E'')$.  By way of contradiction, suppose that there exists an edge $e=uv\in E'\setminus E''.$  Let $P$ be a shortest path in $G''$ connecting $u$ and $v$ (and note that $e\notin P$).  Consider the last edge in $P\cup\{e\}$, say $\tilde{e}$, examined by the greedy algorithm when forming $G'$.  Since the edges are sorted, we have $w_e\leq w_{\tilde{e}}$.  Since $P\cup \{e\}$ lies in $G'$, and each of these edges has weight at most that of $\tilde{e}$, it follows that all edges in $P\cup \{e\}\setminus\{\tilde{e}\}$ have already been added to $E'$ by the time the greedy algorithm examines $\tilde{e}$. Thus this set forms a path connecting the endpoints of $\tilde{e}$, and we have
\[ W(P)-w_{\tilde{e}}+w_e\leq W(P)\leq tw_e\leq tw_{\tilde{e}}, \]
which implies that the greedy algorithm will not add edge $\tilde{e}$ to $E'$, which is a contradiction.  Hence no such $e$ exists, and $E''=E'$.
\end{proof}

\begin{theorem}[{\cite[Theorem 4]{filtser2016greedy}}]\label{THM:GreedyExistentialOptimal}
Suppose that $\mathcal{G}$ is a class of graphs on $n$ vertices which is closed under edge deletion.  Let $t>1$ be fixed.  Let
\[ m(t,n):=\sup_{G\in\mathcal{G}}\;\underset{G'\text{ is a }t\text{-spanner for }G}{\underset{G'=(V,E'),}{\inf}}\;|E'|,\] and
\[ \ell(t,n):= \sup_{G\in\mathcal{G}}\;\underset{G'\text{ is a }t\text{-spanner for }G}{\underset{G'=(V,E'),}{\inf}}\;W(E').\]
Then for every $G\in\mathcal{G}$, the greedy $t$--spanner, $G'$ of $G$ has at most $m(t,n)$ edges and weight at most $\ell(t,n)$.
\end{theorem}
\begin{proof}
The key ingredient is that since a $t$--spanner of $G\in\mathcal{G}$ can be obtained by edge deletion, it must be in $\mathcal{G}$ as well.  Combining this observation with Lemma \ref{LEM:Greedy} yields the desired bounds immediately.  Indeed, let $G\in\mathcal{G}$ be arbitrary, and let $G'$ be its greedy $t$--spanner.  Since $G'$ can be obtained from $G$ by edge deletion, it is in $\mathcal{G}$; hence by assumption $G'$ has a $t$--spanner, say $G''$, which has at most $m(t,n)$ edges and weight at most $\ell(t,n)$.  But since $G''=G'$, the proof is complete. 
\end{proof}

 For a proof of optimality of the greedy algorithm for geometric graphs, see~\cite{borradaile2019greedy}.


\subsection{Arbitrarily Bad Greedy $t$--spanners}

Universal optimality of course implies existential optimality, but the reverse is patently untrue.  Here, we note that not only is the greedy $t$--spanner not universally optimal, but moreover there is a whole family of graphs for which the greedy $t$--spanner is as far away as possible from the optimal $t$--spanner.  This family of examples is inspired by the one given by Filtser and Solomon \cite{filtser2016greedy} based on the Petersen graph.  Consider a complete bipartite graph on $n$ edges (seen in Figure \ref{FIG:BipartiteGreedy}) where each edge has weight 1.  Subsequently, we add an extra vertex, which is connected to every vertex in the original bipartite graph by an edge with weight $1+\varepsilon$.  Suppose that $2<t<3$ is fixed; then if $\varepsilon$ is suitably small, the greedy $t$--spanner of $G$ is the the bipartite graph and half of the edges connecting the additional vertex.  On the other hand, the optimal $t$--spanner is the {\color{black}star} graph that sits atop of the original bipartite graph.

\begin{figure}[!h]
\centering
\includegraphics[trim={0 0.2in 0 0},clip,width=.3\textwidth]{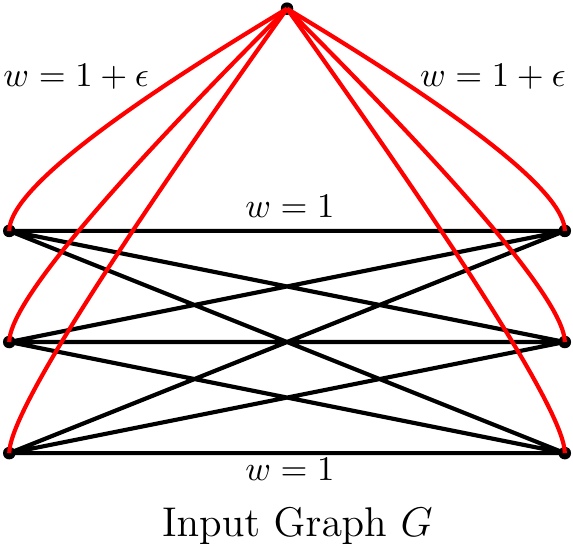}\qquad
\includegraphics[trim={0 0.2in 0 0},clip,width=.3\textwidth]{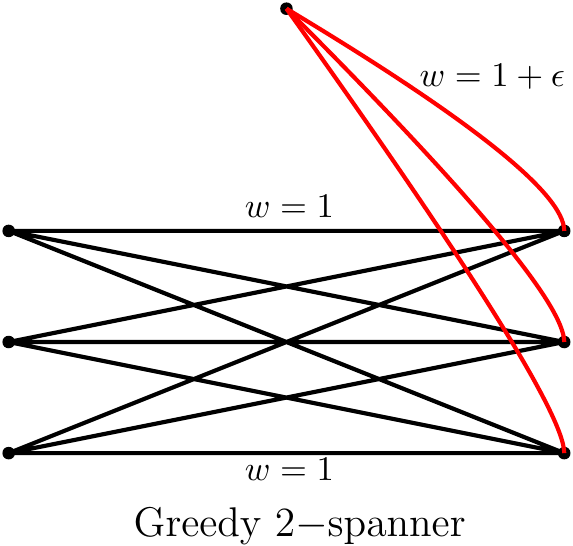}\qquad
\includegraphics[trim={0 0.2in 0 0},clip,width=.3\textwidth]{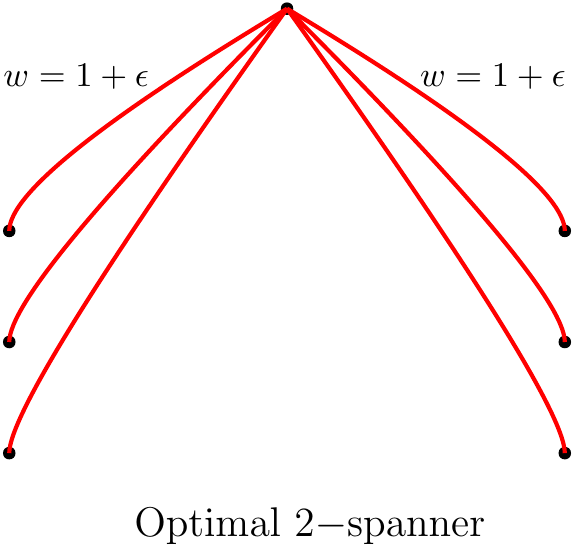}

\caption{(\textbf{Top Left}) Input Graph $G$, (\textbf{Top Right}) The Greedy $t$--spanner for any $2<t<3$, (\textbf{Bottom}) the optimal $t$--spanner for any $2<t<3$.}\label{FIG:BipartiteGreedy}
\end{figure}

In this case, if $G'=(V,E')$ is the greedy $t$--spanner, and $G_{\text{OPT}}=(V,E_{\text{OPT}})$ is the optimal $t$--spanner, we find that
\[ |E_{\text{OPT}}| = n,\quad |E'| = n^2+n, \]
and
\[ W(E_{\text{OPT}}) = n(1+\varepsilon), \quad W(E') = n^2+n(1+\varepsilon).\]
In other words,
\[\dfrac{|E'|}{|E_{\text{OPT}}|} = \Theta(n),\quad  \dfrac{W(E')}{W(E_{\text{OPT}})} = \Theta(n).\]  Thus, greedy $t$--spanners can be arbitrarily worse in terms of both sparsity and lightness than optimal $t$--spanners.

\subsection{Sparsity and Lightness Guarantees for Greedy $t$--spanners}

\begin{itemize}
\setlength{\itemindent}{.2in}
\item[\cite{chandra1992new}] \bibentry{chandra1992new}
\item[\cite{elkin2014light}]\bibentry{elkin2014light}
{\color{black}\item[\cite{elkin2016fast}]\bibentry{elkin2016fast}}
\item[\cite{chechik2018near}] \bibentry{chechik2018near}
\end{itemize}

Implicit in the formulation of the $t$--spanner problem is the idea that a good $t$--spanner for a graph should contain very few edges while still preserving distances in the manner prescribed.  While Algorithm \ref{ALG:Greedy t spanner} does not require the greedy ordering of the edges to produce a $t$--spanner, this ordering is crucial to constructing a sparse $t$--spanner, i.e. one that has few edges and/or small total edge weight.  The original bounds given by Alth\"{o}fer et al.~are the following.
\begin{theorem}[{\cite[Theorem 1]{althofer1993sparse}}]\label{THM:Althoferweak}
Let $G' =(V,E') = \textsc{GREEDYSPANNER}(G,2t+1)$ where $t>0$. Then $G'$ is a $(2t+1)$--spanner of $G$ and
\begin{enumerate}
    \item $|E'| < n \ceiling{n^\frac1t}$, 
    \item $W(E') < (1+\frac{n}{2t})W(\MST(G))$.
\end{enumerate}
\end{theorem}

However, in modern papers a slightly strengthened version of this theorem is usually quoted:

\begin{theorem}[{\cite[Theorem 1, Strengthened]{althofer1993sparse}}]\label{THM:Althofer}
Let $k$ be a positive integer and let $G' = (V,E') = \textsc{GREEDYSPANNER}(G,2k-1)$.  Then $G'$ is a $(2k-1)$--spanner of $G$ and
\begin{enumerate}
    \item \label{ITEM:greedyedges} $|E'| = O\left(n^{1+\frac1k}\right)$, 
    \item $W(E') = O(1 + \frac{n}{k})W(\MST(G))$.
\end{enumerate}
\end{theorem}
\begin{proof} [Proof of Theorem \ref{THM:Althofer}, (\ref{ITEM:greedyedges})]
As discussed in Section \ref{SEC:greedymult} above, the greedy spanner with parameter $2k-1$ has no cycles on $\le 2k$ edges; that is, it has girth $>2k$.
By the Moore Bounds (Proposition \ref{PROP:Moore}), it thus has $O\left(n^{1+\frac{1}{k}}\right)$ edges.
\end{proof}

An important observation implicit in the stronger phrasing is that it is without loss of generality to consider only \emph{odd integer} stretch parameters for multiplicative spanners, at least with respect to extremal spanner size.
This essentially follows from two graph-theoretic facts.
First is that (as discussed in Section \ref{SEC:greedymult}) the extremal sparsity of a $t$--spanner is the same as the extremal sparsity $\gamma(n, t)$ of a graph with girth $>t$, and girth is an integer parameter.
Thus we have (say) $\gamma(n, 5) = \gamma(n, 5.5)$, and so the extremal sparsity of a $5$--spanner is the same as the extremal sparsity of a $5.5$--spanner, even though the latter is strictly more accurate.
Thus the size bound for $5$--spanners is a strictly stronger result than the size bound for $5.5$--spanners.

The second fact is that $\gamma(n, 2k) = {\color{black}\Theta(\gamma(n, 2k+1))}$.
This holds for the following reason: given a graph $G$ with girth $>2k$, one can find a bipartite subgraph $H$ by placing each node on the left or right side of the bipartition with probability $\frac{1}{2}$, and then keeping only edges crossing the divide.
One computes that each edge survives in $H$ with probability $\frac{1}{2}$, and thus (in expectation) the density of $H$ is within a constant fraction of the density of $G$.
On the other hand, all cycles in $H$ are even, so $H$ has girth $>2k+1$.
Together, these facts that the extremal size of (say) a $5$--spanner is the same as the extremal size of a $6.99$--spanner, up to constant factors.
Thus if we hide constant factors, as in Theorem \ref{THM:Althofer}, the size bounds for odd integer stretch subsumes all other possible stretch values.

 The size bound in Theorem \ref{THM:Althofer} is tight assuming the Girth Conjecture; however, the weight bound is not.  Indeed, Chandra et al.~\cite{chandra1992new} improved the bound to $O_\epsilon(k\cdot n^\frac{1}{k})$ for the greedy $(2k-1)(1+\epsilon)$--spanner. Elkin et al.~\cite{elkin2014light} further improved the bound to $O_\epsilon(\frac{k}{\log k}\cdot n^\frac{1}{k})${\color{black}; the construction time of this result was then improved in~\cite{elkin2016fast}}. Recently, Chechik and Wulff-Nilsen \cite{chechik2018near} proposed a new spanner construction with lightness $O_\epsilon(n^\frac{1}{k})$, completely removing dependency of the factor on $k$. According to the existential optimality of the greedy spanner as shown by Filtser and Solomon~\cite{filtser2016greedy}, the same sparseness bound holds for the original greedy spanner algorithm. This result is optimal up to a $(1+\epsilon)$ factor in the stretch provided the girth conjecture is true.

\subsection{Summary of Greedy Algorithm Guarantees}

\begin{table}[H]
     \centering
     \begin{tabular}{|c|c|c|c|c|}
     \hline
         Stretch ($t$)  & Size: $O(|E'|)$ & Weight: $O(W(E'))$ & Time & Ref.  \\
         \hline\hline
          $2k-1$  & $n^{1+\frac1k}$ & $(1+\frac{n}{2k})W(\textnormal{MST}(G))$ & $m(n^{1+\frac1k}+n\log n)$ & \cite{althofer1993sparse}\\
         \hline
         $(2k-1)(1+\eps)$  & $n^{1+\frac1k}$ & $kn^\frac1k \left(\frac1\eps\right)^{1+\frac1k}$ & $m(n^{1+\frac1k}+n\log n)$ & \cite{chandra1992new}\\
         \hline
          $(2k-1)(1+\eps)$  & $n^{1+\frac1k}$ & $kn^\frac1k \left(\frac1\eps\right)^{1+\frac1k}$ & $n^\frac1k\left(1+\frac{k}{\eps^{1+\frac1k}\log k}\right)$ & \cite{elkin2014light}\\
         \hline
     \end{tabular}
     \caption{Spanners given by the Greedy Algorithm. Here $n=|V|$ and $m=|E|$. }
     \label{TAB:Greedy}
 \end{table}

\subsection{Open Questions}

\begin{enumerate}
\item Is there a natural reverse greedy algorithm for constructing a $t$--spanner?  That is, one which begins with the full edge set $E$ ordered in nonincreasing order, and deletes edges according to some criteria to arrive at a $t$--spanner.

\item Is there a greedy algorithm to produce additive spanners, $(\alpha,\beta)$--spanners, subsetwise spanners, or more generally pairwise spanners?
(The additive spanner construction of \cite{Knudsen14} has a greedy step; can it be brought completely in line with the multiplicative greedy algorithm, or understood by a similar analysis?)
\end{enumerate}

\section{Clustering and Path Buying Methods}\label{SEC:Cluster}

Another prominent set of techniques for computing graph spanners is what we call \emph{clustering and path buying} techniques.  In the clustering phase, one starts from a set of vertices of a given graph as initial clusters and then expand each to produce a clustering of the graph.  Typically at the same time, we add edges within the clusters to the candidate spanner.  In the path buying phase, we then add (buy) cheap edges to achieve the desired spanner.  Let us stress here that the clustering we are discussing is much different than common graph clustering algorithms such as $k$--means or spectral clustering whose aim is to partition the graph into disjoint clusters.  

\subsection{General Approach}
Here we provide a general definition for clustering in the context of spanners. 
\begin{definition}

A clustering $\mathscr{C}=\{ C_1, \dots , C_k \}$ of the graph $G=(V,E)$ is a collection of clusters $C_i \subseteq V$.
\end{definition} 

While this definition is quite general, typically one imposes extra conditions on the clustering based on the specific problem at hand.  For instance, one may require
$ \bigcup_i C_i = V$, or one may require that clusters are pairwise disjoint. 

Intuitively, the clustering step aims to attack the spanner problem locally, and the final step makes global adjustments to find the appropriate spanner.  As an example of this global adjustment, one can iterate through a given set of vertex  pairs $(u,v)$  and make sure there exists a $u$-$v$ path with the desired stretch factor.

To better illustrate common variants of this set of techniques, we first list some useful general concepts to keep in mind. Variants of each of these will be used in the different algorithms discussed in the sequel.
\begin{itemize}\normalsize
\item \textbf{Value of a Path}: The value of the path is the overall improvement of the stretch factors that we gain by adding the path to the spanner under construction. For example, for a given path $\rho$ we can measure the number of clusters that will be closer together after adding $\rho$ to the spanner. 

\item \textbf{Cost of a Path}: In the path buying phase of the algorithm, we may decide to add some edges to the spanner to maintain a given path in $G$. The cost of a path typically is the weight of the new edge set that we are adding to the spanner.



\end{itemize}

Most clustering and path buying algorithms run in two stages, which we summarize in the following proto-algorithms.  Typical clustering algorithms begin with singleton vertices as clusters, which are then grown according to some rule which we simply call Rule 1, and subsequently edges are added to the spanner based on different criteria (called Rule 2) depending on if they connect vertices within a cluster or not. 

\begin{algorithm}[H]
\caption{\textsc{Proto-Clustering}($G$, Rule 1, Rule 2)}
\label{ALG:ProtoClustering}
\begin{algorithmic}
\State Initialize $E'\gets \emptyset$
\State Choose initial clusters $C_1,\dots C_k$ (deterministically or randomly)
\For{$v\in V$}
\State Add nearby vertices to a cluster, and add nearby edges to $E'$ according to Rule 1
\EndFor
\State If an edge is not nearby a cluster, add it to $E'$ according to Rule 2\\
\Return $E'$
\end{algorithmic}
\end{algorithm}

To specify a proto-algorithm for the path buying stage, one first chooses a notion of the cost and value of a path which will be denoted cost$(\rho)$ and value($\rho$) for a given path $\rho$, respectively.  With these as parameters, the path-buying proto-algorithm may be stated as follows (the set of all paths in $G$ will be denoted by $\PP$).

\begin{algorithm}[H]
\caption{\textsc{Proto-Path Buying}$(G,\textnormal{cost}(\rho),\textnormal{value}(\rho),\alpha$, Rule 1, Rule 2)}
\label{ALG:ProtoPathBuying}
\begin{algorithmic}
\State Initialize $E'\gets $\textsc{Proto-Clustering}($G$, Rule 1, Rule 2)
\For{$\rho\in\PP$}
\If{$\textnormal{cost}(\rho)\leq\alpha\cdot\textnormal{value}(\rho)$}
\State $E'\gets E'\cup \{e\in\rho\}$
\EndIf
\EndFor\\
\Return $E'$
\end{algorithmic}
\end{algorithm}

The proto-algorithms given here have many degrees of flexibility: namely the initialization of clusters, the rules for adding clustered and unclustered edges to $E'$, the definitions of cost and value of a path, and the relation of the final two quantities.  As a general rule, the clustering phase is the cheap one in terms of run-time complexity, whereas the path buying phase is more expensive due to the fact that one typically runs through all possible paths in $G$.

Let us also note that several of the clustering algorithms utilize a random edge selection step; however, the guarantees for the resulting spanners are deterministic.  For algorithms which produce spanners only with high probability, see Section \ref{SEC:Probabilistic}.

\subsection{Illustrating Example}\label{SEC:ClusterExample}

\begin{itemize}
\setlength{\itemindent}{.2in}
    \item[\cite{baswana2010additive}] \bibentry{baswana2010additive}
\end{itemize}

Let us begin with one of the simpler algorithms in the vein described above given by Baswana et al.~\cite{baswana2010additive} which computes a $(1,6,V\times V)$--spanner (i.e. an additive 6--spanner) of an unweighted graph.  We describe only the terminology and rules required to state their algorithm in terms of the proto-algorithms of the previous section.

\textbf{Initialization of Clusters:}  first, $|V|^\frac23$ cluster centers are chosen uniformly randomly from $V$.  That is, we have $C_1 = \{v_1\}, \dots, C_{k}=\{v_k\}$ for $k=|V|^\frac23$.

\textbf{Rule 1:} if $v$ is adjacent to a cluster center, it joins an arbitrary cluster that it is adjacent to, and the corresponding edge is added to $E'$.

\textbf{Rule 2:} if $v$ has no adjacent cluster center, then it is left unclustered and all edges incident to $v$ are added to $E'$.

For the path-buying phase of this algorithm, the parameters involved are defined as follows:

\begin{equation}\label{EQ:Cost} \textnormal{cost}(\rho):= |\rho \setminus E'|,\end{equation}
that is, the number of edges in $\rho$ that are not already in the spanner, and
\begin{align*}
\textnormal{value}(\rho)&:=|\{(C_1,C_2)\in\mathscr{C}\times\mathscr{C}: \rho \textnormal{ intersects } C_1,C_2 \textnormal{ and } \\ &\textnormal{dist}_{G'}(C_1,C_2) \textnormal{ decreases after adding } \rho \textnormal{ to } G'\} \\
\alpha &:= 2
\end{align*}
where $G'$ is the current spanner and $\textnormal{dist}_{G'}(C_1,C_2)$ denotes the length of the shortest path in $G'$ between terminal vertices that lie in $C_1, C_2$ respectively.  

\begin{theorem}[{\cite[Theorem 2.7]{baswana2010additive}}]\label{THM:16Spanner}
Given the cluster initialization, Rule 1, 2, and definitions of $\alpha$, cost and value of paths above, the algorithm \textnormal{PROTO-PATH BUYING} returns a $(1,6,V\times V)$--spanner of $G$.
\end{theorem}
\begin{proof}[Sketch of Proof]
For the sake of illustration, we sketch the proof that the algorithm described in Theorem \ref{THM:16Spanner} produces a $(1,6)$--spanner. Consider the special case of two clustered nodes $u$ and $v$. For a shortest path $\rho$ from $u$ to $v$, we look at a sequence of clusters $C_1,\dots C_\ell$ intersecting $\rho$ with $u\in C_1$ and $v\in C_\ell$.
We will prove that after the path-buying phase, the spanner $G'$ satisfies the following property.

\textbf{ Intermediate Cluster Property (ICP)}: A shortest path $\rho$ satisfies the ICP if there exists a cluster  $C_i\,$ ($1< i<\ell$) along $\rho$ such that
\[
\textnormal{dist}_{G'}(C_1,C_i) \leq \textnormal{dist}_{\rho}(C_1,C_i),\quad \textnormal{dist}_{G'}(C_i,C_\ell) \leq \textnormal{dist}_{\rho}(C_i,C_\ell).
\]
Here, $\textnormal{dist}_{\rho}(C_i,C_j)$ is the distance from $C_i$ to $C_j$ along the prescribed path $\rho$. Note that in the description of Rule 2, $v$ joins only one of the neighbor clusters and hence there are two types of missing edges in $\rho$: those with both endpoints in the same cluster (intracluster) and those with endpoints belonging to different clusters (intercluster).  Now, assuming that the ICP is true for some shortest path $\rho$ and counting intracluster and intercluster missing edges, we can apply  the triangle inequality and use the fact that the diameter of each cluster is two, to get the desired $(1,6)$--spanner condition.


Now, our goal is to prove the ICP; to do so, we first define the following sets:
\[
A = \{ \{C_i,C_j\}  :  i=1 \, or \, j=\ell \}
\]
\[
A_0 = \{ \{C_i,C_j\} \in A  :  \textnormal{dist}_{G'}(C_i,C_j) > \textnormal{dist}_\rho(C_i,C_j) \},
\]
\[
A_1 = A \setminus A_0.
\]
Recall the definition for the value of a path:
\[
\textnormal{value}(\rho) = | \{ (C_i,C_j) :  \textnormal{dist}_{G'}(C_i,C_j) > \textnormal{dist}_{\rho}(C_i,C_j) \} |,
\]
hence we have $|A_0| \leq \textnormal{value}(\rho)$. 

\textbf{Claim:} $\val(\rho)<\ell -2$. 

If the claim is true then $|A_1| > \ell-1$, and using the pigeonhole principle there exists some $C_i$ satisfying the ICP. It remains to prove the claim. With a simple counting argument, and if shortest path ties are broken properly (roughly, shortest paths should stay in their current cluster as long as possible), given that $\ell$ is the number of clusters intersecting $\rho$, the number of missing intercluster edges in $G'$ is at most $\ell-1$ and the number of missing intracluster edges is at most $\ell -2 $. Therefore, $\cost(\rho)<2\ell -3$ and using the fact that $2 \val(\rho) < \cost(\rho)$,
we get $\val(\rho)<\ell-2$ which finishes the proof.
\end{proof}

\subsection{Theme and Variations}
\begin{itemize}
\setlength{\itemindent}{.2in}
\item[\cite{Surender03ALP}] \bibentry{Surender03ALP}
\item[\cite{pettie2009low}] \bibentry{pettie2009low}
\item[\cite{cygan13}] \bibentry{cygan13}
\item[\cite{kavitha2013small}] \bibentry{kavitha2013small}
\item[\cite{kavitha2017new}] \bibentry{kavitha2017new}
\end{itemize}


Cygan et al.~\cite{cygan13} give a polynomial time clustering and path buying algorithm for computing $(\alpha,\beta)$ pairwise, subsetwise, and sourcewise spanners.
Their clustering step is essentially a deterministic version of the one used by Baswana et al.~\cite{baswana2010additive}; this derandomization is orthogonal to the change from all-pairs to pairwise spanners, but we will include it here for completeness.

\subsubsection*{Clustering Phase}

This follows Algorithm \ref{ALG:ProtoClustering} with the following rules.
In the sequel, let $d$ be an integer parameter of the construction that we will choose later.

\textbf{Choice of Clusters:} Unmark all nodes, and then while there is a (possibly marked) node $v \in V$ with $\ge d$ unmarked neighbors, choose a set $C$ of exactly $d$ of its neighbors and add $C$ as a new cluster.
The node $v$ is called its \emph{center} (note that $v \notin C$).
We then mark all nodes in $C$, and repeat until we can do so no longer.

\textbf{Rule 1:} All edges with both endpoints contained in the same cluster are included in the spanner.  All edges between a cluster center and a node in its cluster are included in the spanner.

\textbf{Rule 2:} For each node $v$ not contained in a cluster, include all of its incident edges in the spanner.




The essential properties of this clustering step are as follows.

\begin{lemma}[\cite{cygan13}]
Algorithm \ref{ALG:ProtoClustering}, with the above parameters, provides a clustering $\mathscr{C}$ of $G$ and a subgraph $G'=(V,E')$ with the following properties:
\begin{enumerate}
    \item\label{ITEM:ClusterSize} The size of each cluster $C \in \mathscr{C}$ is exactly $|C| = d$.
    \item \label{ITEM:Clusters} The total number of clusters is $|\mathscr{C}| \le n/d$.
\item \label{ITEM:Neighbor} Any two nodes in the same cluster have a common neighbor in $G'$ and hence the diameter of any cluster is at most 2. 
\item\label{ITEM:SizeE} $|E'| = O(nd)$
\item  \label{ITEM:Absent} If an edge $uv$ is absent in $E'$, then $u$ and $v$ belong to two different clusters (In particular, they cannot be unclustered nodes).
\end{enumerate}
\end{lemma}
\begin{proof}
Items (\ref{ITEM:ClusterSize}), (\ref{ITEM:Neighbor}) and (\ref{ITEM:Absent}) follow directly from the algorithm. (\ref{ITEM:Clusters}) follows directly from (\ref{ITEM:ClusterSize}) since clusters are node-disjoint.
Finally, we argue (\ref{ITEM:SizeE}) as follows.
In Rule 1, we add at most $d \choose 2$ edges per cluster, which is $O(nd)$ edges in total.
For each edge $uv$ added in Rule 2, we note that one (or both) endpoints are unmarked.
Since every node has $<d$ unmarked neighbors, by a union bound over the $n$ nodes we add $O(nd)$ edges in this step as well.
\end{proof}

\subsubsection*{Path Buying Phase}

Cygan et al.~give several variants of the path buying phase to obtain different types of spanners.  All begin by using the above Clustering algorithm, but vary the definition of value of paths and $\alpha$ as in the Proto-path buying algorithm.  The first variant is to define cost as in \eqref{EQ:Cost}, and value as
\begin{multline}\label{EQ:Path12} 
\textnormal{value}(\rho) :=|\{(x,C)\in S\times\mathscr{C}: \rho \textnormal{ intersects } x, C, \textnormal{ and } \textnormal{dist}_{G'}(x,C) \\ \textnormal{ decreases after adding } \rho \textnormal{ to } G'\}|.
\end{multline} 
With these definitions of cost and value of a path, the clustering + path-buying algorithm yields an additive subsetwise spanner with additive stretch 2 by setting $\alpha=2$.

\begin{theorem}[{\cite[Theorem 1.3]{cygan13}}; also \cite{elkincomm, pettie2009low}]\label{add-sourcewise}
 For any $ S \subseteq V $, given the above cluster initialization with parameter $d = \sqrt{|S|}$, cost and value of paths specified by \eqref{EQ:Cost}, \eqref{EQ:Path12}, and $\alpha=2$, the algorithm \textsc{PROTO-PATH BUYING} returns a $(1,2,S\times S)$--spanner of $G$ of size $O(n\sqrt{|S|})$.
\end{theorem}

We remark that any choice of $d$ in the \textsc{CLUSTER} step will produce a $(1,2,S\times S)$--spanner, but the choice of parameter $d=\sqrt{|S|}$ is needed to minimize the size of the final spanner.
In particular, the clustering phase costs $O(nd)$ edges and the path buying step costs $O(|S| n / d)$ edges, and these balance at the choice $d = \sqrt{|S|}$.

The second variant of the path-buying algorithm uses path-buying as a \emph{preprocess} in a larger spanner construction.
This is more complicated to prove, but the statement is as follows:

\begin{theorem}[{\cite[Theorem 1.1]{cygan13}}]
\label{add-pairwise}
 For any $ \eps > 0 $, positive integer $k$, and any $P \subseteq V\times V$, given the above cluster initialization with parameter
 $$d = n^\frac1k(2k+3)|S|^\frac{k}{2k+1},$$
 cost and value of a path specified in Section \ref{SEC:ClusterExample}, and
 $$\alpha = \frac{12\log n}{\eps},$$
 let the edge set of the output of the algorithm \textsc{PROTO-PATH BUYING} be $E_1$.  Let $E_2$ be the edges of a $(2\log n,0,V\times V)$--spanner of $G$ having size $|E_2|=O(n)$ \cite{Halperin96}.  Then $G'=(V,E_1\cup E_2)$ is a $(1+\eps,4,P)$--spanner of $G$.
\end{theorem}
   
The final variant involves the same clustering step, but a more complicated path-buying phase. For full details, see \cite{cygan13}; details of their other results are in the tables in the Appendix.  

Baswana et al.~\cite{baswana2010additive} also give another variant of their algorithm which computes a $(k,k-1,V\times V)$--spanner.  Rather than force it into our proto-algorithm framework, we reproduce it in full here in Algorithm \ref{alg:kk-1spannerconstruction}.
The clustering used here is related to one in \cite{thorup2006spanners}.

\begin{algorithm}[!ht]
\caption{Randomized ($k,k-1$)--Spanner Construction($G$)} \label{alg:kk-1spannerconstruction}
\begin{algorithmic}
\State Initialize $E' \gets \emptyset$ and $\mathscr{C}_0 \gets \{\{v\}:v \in V\}$
\For{$i$ from 1 to $k$}
\State Let $\mathscr{C}_i$ be sampled from $\mathscr{C}_{i-1}$ with probability $n^{-\frac{1}{k}}$ (If $i=k$, then $\mathscr{C}_k = \emptyset$)
\For{$v$ not belonging to a cluster in $\mathscr{C}_i$}
\State If $v$ is adjacent to some $C \in \mathscr{C}_i$, add $v$ to $C$ and add some edge of $E(v,C)$ to $E'$
\State Otherwise, add to $E'$ some edge from $E(v,C)$, for each $C \in \mathscr{C}_{i-1}$ adjacent to $v$
\EndFor
\EndFor
\State Add to $E'$ one edge from $E(C,C')$ for each adjacent pair $C \in \mathscr{C}_i$ and $C' \in \mathscr{C}_{k-1-i}$ for $i$ from 0 to $k-1$
\State Add to $E'$ one edge from $E(C,C')$ for each adjacent pair $C \in \mathscr{C}_i$ and $C' \in \mathscr{C}_{i-1}$ for $i$ from $\lceil \frac{k}{2} \rceil$ to $k-1$
\State \textbf{return} $E'$
\end{algorithmic}
\end{algorithm}

Kavitha and Varma \cite{kavitha2013small} give yet another clustering and path buying algorithm for computing $(\alpha,\beta,P)$--spanners which utilizes \textit{breadth-first search (BFS) trees}~\cite{clrs} in the path-buying phase.   Recently, Kavitha \cite{kavitha2017new} presents a modified version of the BFS strategy to compute additive pairwise spanners.  The broad idea of both is that BFS trees provide highways along which distances are preserved, and they are used to connect clustered nodes together.

Kavitha~\cite{kavitha2017new} has provided algorithms similar in spirit to those of Cygan et al.~to compute additive pairwise spanners. 


Baswana and Sen \cite{Surender03ALP} use a similar clustering algorithm to Algorithm \ref{alg:kk-1spannerconstruction} but a different path-buying phase than \cite{baswana2010additive} to yield a $(2k-1)$--spanner with $O(n^{1+\frac1k})$ edges in expected linear time $O(kn)$.  This gives a spanner with optimal number of edges (up to a constant and assuming the Girth Conjecture) in optimal time.

Pettie~\cite{pettie2009low} proposes a {\it modular} scheme for constructing $(\alpha,\beta)$--spanners through utilizing {\it connection schemes}, many of which are variants of this clustering + path buying approach. Through assembling connection schemes properly, most existing results can be generated. Also, substantially improved results for almost additive spanners are obtained.
In particular, it is shown that linear size ($|E'|=O(n)$) spanners can be constructed with good stretch including $(5+\epsilon,\text{polylog}(n))$--spanners, and $(1,\widetilde{O}(n^\frac{9}{16}))$--spanners.
{\color{black} The latter bound was later improved to $(1, O(n^{\frac37 + \eps}))$ in \cite{bodwin2016better}.}

\subsection{Further Reading}

\begin{itemize}
\setlength{\itemindent}{.2in}
    \item[\cite{elkin2004}] \bibentry{elkin2004}
    \item[\cite{BaswanaKMP05}] \bibentry{BaswanaKMP05}
    \item[\cite{baswana2007simple}] \bibentry{baswana2007simple}
    \item[\cite{chechik2013new}] \bibentry{chechik2013new}
\end{itemize}

Baswana et al.~\cite{BaswanaKMP05} give two important results on additive spanners: an additive $6$--spanner of size $O(n^\frac{4}{3})$, and a linear time construction of $(k, k-1)$--spanners with size $O(n^{1+\frac{1}{k}})$. In Baswana and Sen~\cite{baswana2007simple}, the first linear time randomized algorithm that computes a $t$--spanner of a given weighted graph was given. Recall that this size/error tradeoff is optimal assuming the Girth Conjecture (see Section \ref{SEC:girth}).

Chechik~\cite{chechik2013new} gives a construction for an 
additive 4--spanner containing $\widetilde{O}(n^\frac{7}{5})$ edges. In addition, a construction is given for additive spanners with $\widetilde{O}(n^{1+\delta} )$ edges and additive stretch $\widetilde{O}(n^{\frac{1}{2}-\frac{3\delta}{2}} )$ for any $\frac{3}{17} \leq \delta < \frac{1}{3}$, improving the stretch of the existing constructions from $O(n^{1-3\delta} )$ to $\widetilde{O}(\sqrt{n^{1-3\delta}})$. Finally, it is shown that this construction can be modified to give a sublinear additive spanner of size $\widetilde{O}(n^{1+\frac{3}{17}})$ with additive stretch $O( \sqrt{ d(u,v)})$ for each pair $u,v\in V$.

Elkin and Peleg~\cite{elkin2004} show that the multiplicative factor can be made arbitrarily close to $1$ while keeping the spanner size arbitrarily close to $O(n)$ at the cost of allowing the additive term to be a sufficiently large constant. In other words, they show that for any constant $\lambda > 0$ there exists a constant $\beta = \beta(\epsilon, \lambda)$ such that for every $n$--vertex graph there is an efficiently constructible $(1 + \epsilon, \beta)$--spanner of size $O(n^{1+\lambda})$.

\subsection{Open Problems}

\begin{enumerate}
    \item Suppose one runs a traditional graph clustering algorithm (in the unsupervised learning sense; e.g. Normalized Cuts, Spectral Clustering, or Hierarchical Clustering) which produces clusters which are disjoint and cover the whole graph. Can one use this clustering to give rise to a good spanner of $G$?
\end{enumerate}


\section{Probabilistic Methods}\label{SEC:Probabilistic}

\begin{itemize}
\setlength{\itemindent}{.2in}
\item[\cite{miller2015improved}] \bibentry{miller2015improved}
\item[\cite{elkin2018efficient}] \bibentry{elkin2018efficient}
\end{itemize}

Miller et al.~\cite{miller2015improved} construct the first randomized algorithm for producing a spanner with high probability (other algorithms using randomization give deterministic guarantees).  The main result therein (Theorem 1.1) yields a multiplicative $O(k)$--spanner with high probability of expected size $O(n^{1+\frac1k}\log k))$ in expected time $O(m)$.

Elkin and Neiman~\cite{elkin2018efficient} improved on the work of Miller et al.~\cite{miller2015improved} by giving a randomized construction which computes a $(2k-1)$--spanner with $O(n^{1+\frac1k}\eps^{-1})$ edges with probability $1-\eps$.  Moreover, they give a runtime and edge count analysis for both PRAM and CONGEST computational models.  To illustrate this recent technique, we give a partial proof of the main theorem in~\cite{elkin2018efficient}.

\subsection{Probabilistic $(2k-1)$--spanners}

To begin, note that given a parameter $\lambda>0$, the \textit{exponential distribution} is defined by the probability density function (pdf)
\[ p_\lambda(x) = \begin{cases} \lambda e^{-\lambda x} & x\geq0\\ 0 & x<0.\end{cases} \]
We denote a random vector $x$ drawn from this distribution by $x\sim Exp(\lambda)$.

\begin{algorithm}[!htb]
\begin{algorithmic}
\State $\lambda \gets \frac{\ln(cn)}{k}$
\For{$u\in V$}
\State Draw $r_u\sim Exp(\lambda)$
\State $u$ broadcasts $r_u$ to all vertices $v$ within distance $k$
\If{$x$ receives a message from $u$}
\State $x$ stores the value $m_u(x)=r_u-d_{G}(u,x)$
\State $x$ also stores a {\color{black}neighboring} vertex $p_u(x)$ along a shortest path from $u$ to $x$ (if there is more than one possibility, one is selected randomly)
\EndIf
\EndFor
\State $m(x)\gets \max_{u\in V}m_u(x)$
\For{$x\in V$}
\State $E'\gets E'\cup C(x)=\{(x,p_u(x)):m_u(x)\geq m(x)-1\}$
\EndFor
\State\Return $E'$
\end{algorithmic}
\caption{\textsc{Rand Exp} $(2k-1)$--\textsc{spanner}($G,k,c$)}\label{ALG:Prob}
\end{algorithm}

The key reason for using the exponential distribution to determine the messages passed from each vertex is its \textit{memoryless property}, which is the fact that if $x\sim Exp(\lambda)$ and $s,t\in\R$, then $\Prob(x>s+t\;|\;x>t) = \Prob(x>s)$, i.e. does not depend on the value of $t$.  The main result of~\cite{elkin2018efficient} is the following.
\begin{theorem}[{\cite[Theorem 1]{elkin2018efficient}}]
 Let $G=(V,E)$, $k\in\N$, $c>3$, and $\delta>0$ be fixed.  Then with probability at least $(1-\frac1c)\frac{\delta}{1+\delta}$, Algorithm \ref{ALG:Prob} computes a $(2k-1)$--spanner of $G$ which has at most
 \[ (1+\delta)\frac{(cn)^{1+\frac1k}}{c-1}-\delta(n-1)\] edges.
\end{theorem}
\begin{proof}[Sketch of Proof]
Choose $\lambda = \frac{\ln cn}{k}$, and let $X$ be the event $\{\forall u\in V, r_u<k\}$.  Note that $\Prob(r_u\geq k)=e^{-\lambda k} = \frac{1}{cn}$; hence the union bound implies that $\Prob(\exists u\in V \text{ s.t. } r_u\geq k)\leq\frac{1}{c}$.  Hence $\Prob(X)\geq1-\frac1c$.

The key observation made in \cite{elkin2018efficient} is that if the event $X$ holds, then the subgraph $G'=(V,E')$ output by Algorithm \ref{ALG:Prob} is a spanner with stretch at most $2k-1$.  Additionally, since $G'$ is a spanner, it must have at least $n-1$ edges, and so if $Y$ is the random variable $|E'|-(n-1)$, then $Y$ is positive.  Moreover, $\E(Y) \leq n(cn)^\frac1k-(n-1)$ (Lemmas 1 and 2 of \cite{elkin2018efficient}), and so
\[\E(Y\,\vert\, X)\leq\frac{\E(Y)}{\Prob(X)}\leq \frac{1}{1-\frac1c}\left(n(cn)^\frac1k-(n-1)\right). \]
To turn this into a concrete bound, note that Markov's inequality implies that for the given $\delta$,
\[ \Prob(Y\geq(1+\delta)\E(Y\,\vert\,X)\,\vert\,X)\leq\frac{1}{1+\delta}, \]
whereby we have
\[\Prob((Y<(1+\delta)\E(Y\,\vert\,X))\cap X)\geq\left(1-\frac1c\right)\frac{\delta}{1+\delta}. \]
Consequently, if these events hold, then $|E'| = Y+n-1$ is at most $(1+\delta)\E(Y\,\vert\, X)$, which is at most the quantity in the statement of the theorem by combining the above estimates.
\end{proof}




\subsection{Open Problems}
\begin{enumerate}
    \item Can one optimize the construction of Elkin and Neiman over the pdf to get a better guarantee for the size of a $(2k-1)$--spanner?
    \item Can one design a probability distribution over the edges of a graph from which random sampling yields a spanner with high probability?
\end{enumerate}

\section{Subsetwise, Sourcewise, and Pairwise Spanners}\label{SEC:Wise}
\subsection{Pairwise Distance Preservers}\label{SEC:Preservers}

\begin{itemize}
\setlength{\itemindent}{.2in}
\item[\cite{bollobas2005sparse}] \bibentry{bollobas2005sparse}
    \item[\cite{coppersmith2006sparse}] \bibentry{coppersmith2006sparse}
    \item[\cite{bodwin2016better}] \bibentry{bodwin2016better}
\item[\cite{abboud2016error}] \bibentry{abboud2016error}
\item[\cite{bodwin2017linear}] \bibentry{bodwin2017linear}
\item[\cite{Kshitij2017distance}] \bibentry{Kshitij2017distance}
\item[\cite{bodwin2018shortestpath}] \bibentry{bodwin2018shortestpath}
\end{itemize}

A special case of the spanner problem is that of finding \textit{distance preservers}, which requires that distances are preserved exactly between specified pairs of vertices of $G$.  This variant of the problem is important given that some spanner constructions use distance preservers as a subroutine. 
Distance preservers were first introduced and studied in~\cite{bollobas2005sparse} and~\cite{coppersmith2006sparse}. Coppersmith and Elkin~\cite{coppersmith2006sparse} proved the existence of a linear size pairwise distance preserver for any set $P$ of vertex pairs $|P|=O(\sqrt{n})$.  Further work has been done in \cite{bodwin2017linear, bodwin2016better, Kshitij2017distance}.


Recall we may consider subsetwise distance preservers when $P=S\times S$ and sourcewise distance preservers when $P=S\times V$.  A fundamental question about distance preservers is the following:

\begin{problem}
For a fixed $n$ and $P\subseteq V\times V$, what can be said about the number of edges, $|E'|$, in a distance preserver of $G$ over $P$?
\end{problem}

For the moment, consider the simplest case where $P$ has the form $\{s\}\times V$. Here, the distance preserver is a tree with $O(n)$ edges; this fact is known as the Shortest Path Tree Lemma. Note that without further assumptions on the structure of $G$ or $P$, one cannot improve this construction in general.  Indeed, if $G$ is a single path from $u$ to $v$ and $P=\{(u,v)\}$, the preserver must have $n-1$ edges. On the other hand, if $G$ is a clique, there must be $ |P| $ edges in the preserver, and hence in general, $|E'|=\Omega(n+|P|)$.

Here, we present two theorems on the existence and size of distance preservers along with an outline of one of the proofs.  The original result of~\cite{coppersmith2006sparse} is the following.
\begin{theorem}[\cite{coppersmith2006sparse}] \label{THM:CEDP}
 Given an undirected (weighted or unweighted) graph $G$ and a set of vertex pairs $P$, there exists a distance preserver of size $|E'|=O(n+n^\frac{1}{2}|P|)$.
\end{theorem}

The next theorem is a slightly worse in size complexity, but is more general in that it holds for a larger class of graphs.
\begin{theorem}[{\cite{bodwin2017linear}}]\label{THM:BodwinDP}
 Given a (directed or undirected, weighted or unweighted) graph $G$ and a set of node pairs $P$, there exists a distance preserver of size $|E'|=O(n+n^\frac{2}{3}|P|)$.
\end{theorem}

\subsection{Sketch of Proof of Theorem \ref{THM:BodwinDP}}

\begin{definition}[Tiebreaking Scheme]
A \emph{tiebreaking scheme} is a map $\pi$ that sends each vertex pair $(u,v)$ to a shortest path from $u$ to $v$.
\end{definition}

\begin{definition}
 A \emph{consistent tiebreaking} is a tiebreaking scheme $\pi$ such that for every $w,x,y,z\in V$, if $x,y \in \pi(w,z)$ then $\pi(x,y)$ is a subpath of $\pi(w,z)$.
\end{definition}

\begin{lemma}
Every graph has a consistent tiebreaking.
\end{lemma}

\begin{proof}
  Add a small random number to each edge so that we have a unique shortest path between any pair $(u, v)$; one can see that the output is a consistent tiebreaking.
\end{proof}

\begin{definition}[Branching Triple]
A \emph{branching triple} is a set of three distinct directed edges $(u_1,v)$, $(u_1,v)$, $(u_1,v)$ in $E$ that all enter the same vertex.
\end{definition}

Branching triples are based on the closely related notion of \emph{branching events}, used in the proof of Theorem \ref{THM:CEDP}.
We will prove a bound for the edge size of the distance preserver using a bound for the number of branching triples. 
\begin{lemma}
Let $H=(V,\pi(P))$  where $\pi$ is any consistent tiebreaking. Then $H$ has at most  $\binom{|P|}{3}$ branching triples.
\end{lemma}
\begin{proof}
 Assign each edge $e$ to some pair $p =(v,v') \in P$ such that $e \in \pi(p) $ (i.e. edge $e$ belongs to the shortest path from $v$ to $v'$).
 We show that any three pairs in $P$ cannot share two or more branching triples. By way of contradiction, if three pairs $p_1,p_2,p_3$ share two branching triples $(e_1,e_2,e_3)$ and $(e'_1,e'_2,e'_3)$, then a situation similar to the following will always happen: the edge $e_1$ precedes $e'_1$  in $p_1$ and $e_2$ precedes $e'_2$ in $p_2$. Since $\pi$ is a consistent tiebreaking, one can check that $e'_1=e'_2$ which is a contradiction. That means for every three pairs in $P$, there is at most one corresponding branching triple. Consequently the number of branching triples is at most $\binom{|P|}{3}$.
\end{proof}

The first two edges incident to any vertex $v$ will not create a branching triple; after that, each new edge will contribute at least one new branching triple, so a graph with $O(n)$ branching triples has $O(n)$ edges. From the previous lemma, we conclude that if $|P|=O(n^\frac{1}{3})$ then $|E(H)| = O(n)$. Now, given a graph $G$ and a set of vertex pairs $P$, we can partition $P$ into subsets of size $O(n^\frac{1}{3})$ and apply this fact, which gives the desired conclusion.

\subsection{Further Reading}

Bodwin and Williams~\cite{bodwin2016better} proved an upper bound of $O(n^\frac{2}{3}|P|^\frac{2}{3}+n|P|^\frac{1}{3})$ for undirected and unweighted graphs using a new type of tiebreaking scheme.  They also rely on a clustering technique to prove new upper bounds for additive spanners. 

Abboud and Bodwin~\cite{abboud2016error} consider lower bounds for pairwise spanners.  Most importantly, they prove that lower bounds for pairwise distance preservers imply lower bounds for pairwise spanners.
{\color{black}These connections were significantly expanded in~\cite{abboud2017frac}.} 


The existing trivial lower bounds show that in the worst case, the size of a distance preserver is at least linear in $n$ and $p$. Bodwin~\cite{bodwin2017linear} makes some progress on identifying the cases that these bounds are tight, i.e., when a linear size ($|E'|=O(n+p)$) distance preserver is guaranteed. For example, by the result discussed above, if the number of pairs is $O(n^\frac{1}{3})$ then one can always find a distance preserver with $O(n)$ edges.

Abboud and Bodwin \cite{AB18} introduce and study \emph{reachability preservers}, in which only reachability (typically in a directed graph) rather than distances between demand pairs must be preserved.

\subsection{Multiplicative Pairwise Spanners}
\begin{itemize}
\setlength{\itemindent}{.2in}
\item[\cite{elkin2017terminal}] \bibentry{elkin2017terminal}
\end{itemize}

The paper \cite{elkin2017terminal} consists of various embedding techniques of a metric space (or graph) into a normed space (or a family of graphs), bounding metric distortion for connecting nodes to  terminals. 
As an example of the generic results in the paper, if one can embed any $m$--node graph into $\ell_p$ with distortion $\alpha(m)$ using $\gamma(m)$ dimensions, then there is an algorithm for embedding any graph with $k$ terminals into $\ell_p$, with distortion $2\alpha(k)+1$, using $\gamma(k)+n-1$ dimensions.

A \emph{terminal graph spanner} is another name for a sourcewise spanner; that is, given a graph $G = (V, E)$ and a set of ``terminal nodes'' $K \subseteq V$, the spanner property must hold for all pairs in $K \times V$.
This name is more common when dealing with metric embeddings rather than general graphs.
Another result of \cite{elkin2017terminal} is the following construction of a $(4t-1)$--terminal spanner with $O(n+\sqrt{n} |K|^{1+\frac{1}{t}})$ edges, for a graph $G$ with terminal set $K\subseteq G$:

\begin{itemize}\normalsize
    \item Create a $(2t-1)$ metric spanner $H'$ of $G$, with $P=\mathcal{E}(H')$ and $|P| < O(|K|^{1+\frac{1}{t}})$
    
    \item Applying a distance preserver algorithm to $P$, obtain $G' \subseteq G $ with $O(n+\sqrt{ n} |P|)$ edges, preserving distances in $K$
    
    \item Create $H\supset G'$ by adding shortest path tree in $G$ with $K$ as root. 
\end{itemize}

One reason that pairwise spanners with multiplicative error are perhaps less well studied than others types of error is that they are nontrivial only in a restricted range of parameters.
Suppose we have $|P| = n^{1 + c}$ demand pairs for some $c > 0$, and we want a multiplicative $t$--spanner.
If $c$ is not too large relative to $t$, then at least $|P|$ edges may be needed in the spanner.
On the other hand, $O(|P|)$ edges always suffice by the following construction (which is folklore).
Preprocess the graph with the all-pairs mixed error spanner of \cite{elkin2004}, with parameters chosen such that the spanner has only $O(|P|)$ edges.
One can then argue that all pairs in $P$ have been well spanned, except maybe for pairs at constant (depending only on $c, t$) distance in the original graph.
Hence, we may add an exact shortest path for all remaining pairs, and these cost only $O(|P|)$ edges in total.

\subsection{Open Problems}
\begin{enumerate}
    \item What is the largest $p = p(n)$ so that, for any $|P| = p$ node pairs in an undirected unweighted $n$--node graph, there is a pairwise spanner of $P$ on $O(n)$ edges with $+c$ additive error?  In particular, it is known that any $|P| = O(n^\frac{1}{2})$ pairs have a distance preserver ($+0$ error) on $O(n)$ edges, but it is not clear if one can handle more pairs with (say) a $+2$ error tolerance.
\end{enumerate}

\section{\color{black} Extremal Bounds for Additive, Mixed, and Sublinear Spanners}\label{SEC:LowerBounds}

\subsection{\color{black} Purely Additive Spanners}

{\color{black}

\begin{itemize}
    \setlength{\itemindent}{.2in}
    \item[\cite{Aingworth99fast}] \bibentry{Aingworth99fast}
\item[\cite{woodruff2006lower}] \bibentry{woodruff2006lower} 
\item[\cite{baswana2010additive}] \bibentry{baswana2010additive}
\item[\cite{chechik2013new}] \bibentry{chechik2013new}
\item[\cite{Knudsen14}] \bibentry{Knudsen14}
\item[\cite{abboud2017frac}] \bibentry{abboud2017frac}
\item[\cite{huang2018lower}] \bibentry{huang2018lower}
\end{itemize}
There are three classic upper bounds results for spanners of unweighted graphs with purely additive error:
\begin{enumerate}
\item Implicit in Aingworth et al.~\cite{Aingworth99fast} is the fact that all graphs have additive $+2$--spanners with $\widetilde{O}(n^{\frac32})$ edges; it has been subsequently observed (in folklore) that the bound can be improved to $O(n^{\frac32})$.
A qualitatively different construction of the same result follows from \cite{baswana2010additive, cygan13, elkincomm, Knudsen14}.

\item Chechik~\cite{chechik2013new} proved that all graphs have additive $+4$--spanners with $\widetilde{O}(n^{\frac75})$ edges.
The size was recently improved slightly to $O(n^{\frac75})$ by Bodwin~\cite{bodwin2020some}.

\item Baswana et al.~\cite{baswana2010additive} proved that all graphs have additive $+6$--spanners with $O(n^{\frac43})$ edges.
A simple reframing of the construction was given by Knudsen~\cite{Knudsen14}, and a qualitatively different construction with slightly worse size $\widetilde{O}(n^{\frac43})$ but faster construction time $\widetilde{O}(n^2)$ was given by Woodruff~\cite{woodruff2010additive}.
Knudsen~\cite{knudsen2017additive} later improved the size to $O(n^{\frac43})$ and the construction time to $O(n^2)$, at the cost of increased error of $+8$.
\end{enumerate}
}

These upper bounds on the size of a spanner versus additive error begin to exhibit a tradeoff curve, where one can construct a sparser additive spanner if more additive error is allowed. One natural question is: can we construct even sparser additive spanners while only incurring a small amount of additive error? Specifically, given $0 < \varepsilon < \frac{1}{3}$, is there a constant $k_{\varepsilon}$ such that every graph $G$ has an additive $+k_{\varepsilon}$--spanner containing $O(n^{1 + \varepsilon})$ edges?
{\color{black}Following some progress in~\cite{abboud2016error, woodruff2006lower},} Abboud and Bodwin~\cite{abboud2017frac} give a negative resolution to this question; they show that for $0 < \varepsilon < \frac{1}{3}$, one cannot even compress an input graph into $O(n^{1+\varepsilon})$ bits, so that one can recover distance information for each pair of vertices within $n^{o(1)}$ additive error. This implies that one cannot construct an additive spanner using $O(n^{1+\varepsilon})$ edges {\color{black} in this parameter regime, unless $\text{poly}(n)$ additive error is allowed.
This construction was subsequently optimized by Huang and Pettie~\cite{huang2018lower}.
In particular, they proved that an additive spanner on $O(n)$ edges must have $+\Omega(n^{\frac{1}{13}})$ error, improving over a lower bound of $+\Omega(n^{\frac{1}{22}})$ from \cite{abboud2016error}.

The reason that all constructions of additive spanners have even error is because there is a reduction to the setting of bipartite input graphs, where error must come in increments of $2$.
The reduction is as follows: given an $n$--node graph $G=(V, E)$, make a graph $G'$ in the following two steps:
\begin{enumerate}
    \item Make two identical copies of the nodes in $V$, called $V_L, V_R$.
    \item For each edge $(u, v) \in V$, add two $(u_L, v_R), (u_R, v_L)$ to $G'$, where $u_L, v_L, u_R, v_R$ are the copies of $u, v$ in $V_L, V_R$, respectively.
\end{enumerate}
Since $G'$ is bipartite, any additive spanner of $G'$ will have only an even amount of error; that is, a $+(2c+1)$--spanner on $G'$ is necessarily also a $+2c$--spanner.
Additionally, a $+2c$--spanner of $G'$ can be converted to a $(+2c)$--spanner of $G$ on at most the same number of edges, by merging together $v_L, v_R$ in the natural way.
To state the consequences of this reduction another way, if one can prove that every $n$--node graph has (say) a $+5$ additive spanner on $\le e^*(n)$ edges, then in fact every graph has a $+4$ additive spanner on $O(e^*(n))$ edges.
Thus, generally, it is not interesting to consider purely-additive spanners with odd error.
Note however that this reduction only applies to \emph{extremal} problems where the goal is to determine the worst-case size/error tradeoff.
For \emph{algorithmic} problems, where the goal is just to find a sparse additive spanner of a particular input graph, the cases of $+2k$ and $+2k+1$ error remain non-equivalent.
Indeed, in \cite{chlamtavc2017approximating}, the authors prove that the algorithmic problem of approximating the sparsest $+1$--spanner of an input graph is already hard (whereas a $+0$--spanner is trivial).
}

\subsection{\color{black} Mixed and Sublinear Spanners}

{\color{black}
\begin{itemize}
    \setlength{\itemindent}{.2in}
    \item[\cite{abboud2018hierarchy}] \bibentry{abboud2018hierarchy}
    \item[\cite{elkin2004}] \bibentry{elkin2004}
\end{itemize}

In light of the previously-mentioned lower bounds of \cite{abboud2017frac}, additive spanners cannot generally achieve near-linear size $O(n^{1+\eps})$.
Thus, it is natural to ask how one \emph{can} achieve spanners of this quality, necessarily paying more than purely-additive error, but hopefully less than purely-multiplicative error.
Indeed, this is possible, and it is studied in two closely related paradigms.
These are:
\begin{enumerate}
    \item \emph{Mixed spanners}, which avoid the above lower bounds by allowing a small $(1+\eps)$ multiplicative stretch in addition to $+\beta$ additive stretch.
    That is, a $(1+\eps, \beta)$ mixed spanner $H$ of an input graph $G$ satisfies
    $$\dist_H(u, v) \le (1+\eps)\dist_G(u, v) + \beta$$
    for all node pairs $u, v$.
    The main construction of mixed spanners was given in \cite{elkin2004}; see also \cite{abboud2018hierarchy, ben2020new} for followup work.
    
    \item \emph{Sublinear spanners}, which have additive error that depends sublinearly on the distance in question.
    That is, an $f(d)$--sublinear spanner $H$ of an input graph $G$ satisfies
    $$\dist_H(u, v) \le \dist_G(u,v) + f(\dist_G(u, v))$$
    for all node pairs $u, v$.
    So, $f(d) = O(1)$ (independent of $d$) corresponds to the case of purely additive spanners.
    The main construction of sublinear spanners was given in \cite{thorup2006spanners}, as well as sublinear \emph{emulators}, which satisfy the same distance requirement but are allowed to have new long-range edges $(u, v)$ of weight $\dist_G(u, v)$ (even though the input graph $G$ is still unweighted). (See Section \ref{SEC:emulators} for more on emulators.)
    Follow-up work on sublinear spanners can be found in \cite{chechik2013new, pettie2009low}.
\end{enumerate}

The upper bounds for mixed spanners from \cite{elkin2004} have a hierarchical nature: for any fixed positive integer $k$ and any $\eps>0$, they can construct $(1+\eps, \beta)$ mixed spanners of size
$$O(\text{poly}(1/\eps) n^{1+\frac{1}{2^{k+1}-1}})$$
with $\beta = O(1/\eps)^{k-1}$.
The Thorup--Zwick \emph{emulators} are similarly hierarchical: for any fixed positive integer $k$, they have size
$$O(n^{1+\frac{1}{2^{k+1}-1}}),$$
and sublinear error function $f(d) = O(d^{1 - \frac1k})$.
This sublinear error function is slightly stronger than $(1+\eps, O(1/\eps)^{k-1})$ mixed error; in fact, it can be viewed as obtaining mixed error of the above type for every possible choice of $\eps$ simultaneously.
These emulators can also be converted to $(1+\eps, \beta)$ spanners of the same quality as the above, by replacing each edge of weight $\le \beta$ with a path through the original graph of length $\le \beta$ (folklore).
It is a major open question in the area to determine if one can build sublinear spanners of the same quality as the Thorup--Zwick emulators without taking on this extra $\text{poly}(1/\eps)$ factor.

On the lower bounds side, it is proved in \cite{abboud2018hierarchy} that the Thorup--Zwick emulator hierarchy is essentially optimal (even for all data structures): for any fixed $k$, one cannot improve the error function to $f(d) = o(d^{1-\frac1k})$, \emph{and} one cannot cannot polynomially improve the size (e.g. to $O(n^{1+\frac{1}{2^{k+1}-1} - 0.001})$).
Thus Elkin and Peleg's mixed spanners are similarly optimal, at least up to their dependence on $\text{poly}(1/\eps)$.

To highlight the surprising hierarchical nature of distance compression revealed by these results, \textit{a priori} it seems reasonable to ask: what kind of spanner size can be achieved with sublinear error function $f(d) = O(d^\frac13)$?
The lower bounds of \cite{abboud2018hierarchy} say that one needs $\Omega(n^{\frac43 - o(1)})$ edges -- but if one is willing to pay $\Omega(n^{\frac43})$ edges, then one can even have $f(d) = O(1)$.
Thus not all sublinear error functions are equally interesting to consider; rather, Thorup--Zwick type sublinear error of the form $f(d) = O(d^{1 - \frac1k})$ is of particular interest.
A similar story holds for the hierarchical tradeoff between $\eps$ and $\beta$ for mixed spanners.

}
 
\subsection{\color{black} Additive and Mixed Spanners of Weighted Graphs}

{\color{black}
\begin{itemize}
    \setlength{\itemindent}{.2in}
    \item[\cite{elkin2019almost}] \bibentry{elkin2019almost}
    \item[\cite{ahmed2020weighted}] \bibentry{ahmed2020weighted}
    \item[\cite{huang2019thorup}] \bibentry{huang2019thorup}
    \item[\cite{elkin2019linear}] \bibentry{elkin2019linear}
\end{itemize}

The above results all concern spanners for \emph{unweighted} graphs.
One can also consider spanners for weighted input graphs.
It will not be possible to have a nontrivial construction of (say) an additive spanner with $+2$ error for any weighted input graph, since one can always scale up the edge weights until no edge can be removed without introducing $>2$ additive error between its endpoints.
Thus it is more natural to look for (say) $+2W$ additive spanners, where $W$ is the maximum edge weight of the input graph.
Purely additive spanners of this type were provided in~\cite{ahmed2020weighted}, and mixed spanners whose additive part scales with $W$ were given in~\cite{elkin2019almost}.
In fact, the latter construction has a stronger property: for each node pair $u, v$, the additive error scales with $W = W(u, v)$ the maximum edge weight along the true $u \leadsto v$ shortest path in the input graph, rather than the maximum edge weight globally.

The bounds for mixed and sublinear spanners/emulators are closely related to the bounds for \emph{hopsets}, a related graph-theoretic object where the goal is to add some long-range weighted edges to a (possibly weighted) graph so that an approximate shortest path between any two nodes can be realized in a small number of hops.
Formal connections between these objects were introduced in \cite{elkin2019linear, huang2019thorup}, and expanded in a recent survey by Elkin and Neiman \cite{elkin2020near}.
}

\subsection{\color{black}Open Problems}

{\color{black}
\begin{enumerate}
    \item Do all undirected unweighted graphs have an additive spanner with $+4$ error on $O(n^{\frac43})$ edges?  Currently, the best known bound for $+4$ additive error is $O(n^{\frac75})$ edges~\cite{bodwin2020some, chechik2013new}; any progress beyond this would be interesting.
    
    \item Can Thorup and Zwick's construction of sublinear emulators~\cite{thorup2006spanners} be converted to spanners with similar size/stretch tradeoffs?
    
    \item How much error is needed for an additive spanner on $O(n)$ edges?  Currently there is a construction of additive spanners on $O_{\eps}(n)$ edges with $+O(n^{\frac37 + \eps})$ error~\cite{bodwin2016better}, and the lower bound is $+\Omega(n^{\frac{1}{13}})$ from~\cite{huang2018lower}.
    
    \item Can the $+6$ additive spanner on $O(n^{\frac43})$ edges from \cite{baswana2010additive} be converted to a $+6W$ spanner of the same size for weighted input graphs with max weight $W$?
    (In \cite{ahmed2020weighted, elkin2019almost}, the corresponding question is resolved affirmatively for the $+2$ and $+4$ additive spanners.
    
\end{enumerate}

}

\section{{\color{black}Linear Programming Formulations}}\label{SEC:ILP}

\begin{itemize}
\setlength{\itemindent}{.2in}
\item[\cite{Sigurd04ESA}] \bibentry{Sigurd04ESA} 
\item[\cite{dinitz2011directed}] \bibentry{dinitz2011directed}
\item[\cite{chlamtac2012}]
\bibentry{chlamtac2012}
\item[\cite{BERMAN201393}]
\bibentry{BERMAN201393}
\item[\cite{chlamtac2016}]
\bibentry{chlamtac2016}
\item[\cite{dinitz2016approximating}]
\bibentry{dinitz2016approximating}
\item[\cite{MLGS_proceeding}] \bibentry{MLGS_proceeding}
\item[\cite{dinitz2019lasserre}]
\bibentry{dinitz2019lasserre}
\end{itemize}

{\color{black}
\subsection{Quick Introduction to Linear Programming and Duality}

(For an in-depth reference on Linear Programming and related topics, see \cite{bertsimas1997introduction}.)
A Linear Programming problem (LP) is an optimization problem of the form  
\begin{equation}\label{eqn:ilp-primal}
\begin{aligned}
\max \quad & c^T x\\
\textrm{s.t.} \quad & Ax \geq b\\
  & x \geq 0    \\
\end{aligned}
\end{equation}
when $c, x \in \mathbb{R}^n$ and $b\in \mathbb{R}^m$ are column vectors, and $A$ is an $m\times n$ matrix.   The number of rows and columns of $A$ are equal to the number of constraints and variables respectively.

If we add the additional constraint that the variables $x_i$ (components of $x$) are integers, the corresponding problem is called an Integer Linear Programming problem (ILP). A 0-1 ILP is an ILP where the variables $x_i$ are restricted to be binary (0 or 1). A \emph{mixed} ILP is one where some variables $x_i$ are required to be integers.

Computing an optimal solution to an ILP is NP--hard, as NP--hard problems can be formulated as instances of an ILP. To obtain a reasonable approximation, one usually constructs a \emph{relaxation} of the ILP problem, where the variables are no longer required to be integral. The \emph{integrality gap} is defined as the maximum ratio between the ILP solution and the solution of the corresponding relaxed LP. 

The most obvious way to obtain an integer solution from a non-integer vector optimizing the relaxation is to simply round to the closest feasible integer vector.
However, there are sometimes far better ways to convert to an integer solution, and more sophisticated tools have been developed to address this rounding.
One such a tool is the \emph{lift and project} method, where we introduce new variables and lift the LP problem to higher dimensions, and then project it back to get a better approximation to the ILP solution.

To illustrate, let's assume that $x_1$ and $x_2$ are binary variables corresponding to paths $P_1,P_2$ in $G$,  with the meaning that if $x_i=1$ we add the path to the spanner and if $x_i=0$ we don't add that path. Say we also want to force the solution to have at least one of the two paths. The obvious way of doing this is by adding a constraint
\[
x_1 + x_2 \geq 1
\]
We are looking for a better alternative to the above constraint. Consider the following non-linear constraint
\[
(1-x_1)(1-x_2) = 0
\]
with $0 \leq x_i \leq 1$. This is a better representation if we want to allow variables to have non-integer values, as it forces one of the variables to take the integer value of one. To overcome non-linearity we have to \emph{lift} the problem and introduce a new variable $x'=x_1 x_2$. We then have the corresponding linear constraint:
\[
0 = (1-x_1)(1-x_2) = 1 - x_1 - x_2 + x_1 x_2 = 1-x_1-x_2+x'
\]
we cannot guarantee that $x'=x_1 x_2$ in the final solution, but we know that this relaxation is not worse than the simple LP relaxation because for $x'\geq 0$,  $1-x_1-x_2+x'=0$ implies $x_1+x_2 \geq 1$. This is the basic idea of the lift and projection method. 

In general, when we are dealing with many such constraints the degree of the corresponding non-linear polynomial can get arbitrarily large. Hence, one can consider an upper bound $d$ for the degree of the corresponding non-linear polynomial. This choice will produce a \emph{hierarchy} of LP relaxations to the original ILP problem. An example of such a method with a  hierarchy of LP relaxations is the \emph{Sherali--Adams} lift and projection method. There is also a slightly stronger version known as \emph{Lasserre} hierarchy which has been applied to spanner problems.

\subsection{An LP Relaxation of the Multiplicative Spanner Problem}\label{SEC:lp-more}
Let us look at a well known LP relaxation of the \emph{directed} $t$--spanner problem, given by Dinitz and Krauthgamer \cite{dinitz2011directed}, where the problem is to compute a $t$--spanner of a strongly connected directed graph with the minimum number of edges. For an edge $(u,v)\in E$, let $\mathcal{P}_{uv}$ denote the set of all $u \to v$ paths whose length is at most $t d_G(u,v)$. The indicator variable $x_e$ is 1 if $e$ belongs to the spanner and 0 otherwise. The variable $f_P$ represents the flow along path $P \in \mathcal{P}_{uv}$.

\begin{align}
\min \quad & \sum_{e\in E} x_e \label{ilp:dinitz2011primal}\\
\textrm{s.t.} \quad & \sum_{p\in P_{uv}: e\in p} f_P \geq x_e \quad  \forall (u,v)\in E,\, \forall e \in E \label{ilp:dinitz2011primal_2}\\
  &  \quad \sum_{P \in P_{uv}} f_P \geq 1   \quad \quad  \forall (u,v)\in E \label{ilp:dinitz2011primal_3}\\
  & \quad x_e \geq 0 \quad \quad \forall e \in E \label{ilp:dinitz2011primal_4}\\
  & \quad f_P \geq 0 \quad \quad \forall (u,v)\in E,\, \forall P \in \mathcal{P}_{uv} \label{ilp:dinitz2011primal_5}\\
\end{align}

If the number of paths is polynomial then the LP problem can be solved optimally, as both the number of variables and constraints will be polynomial. In general, this is not true and we may have an exponential number of variables and constraints at the same time. Here, the idea is to write down the dual LP problem: 

\begin{align}
\max \quad & \sum_{(u,v) \in E} z_{u,v}  \label{ilp:dinitz2011dual}\\
\textrm{s.t.} \quad & \sum_{(u,v) \in E} y_{u,v}^e \leq 1 \quad  \forall e \in E  \label{ilp:dinitz2011dual_1}\\
  &  \quad z_{u,v} - \sum_{e\in E} y_{u,v}^e \leq 0   \quad   \forall (u,v)\in E,\, \forall P \in \mathcal{P}_{uv}  \label{ilp:dinitz2011dual_2}\\
  & \quad z_{u,v} \geq 0 \quad \quad \forall (u,v) \in E  \label{ilp:dinitz2011dual_3}\\
  & \quad y_{u,v}^e \geq 0 \quad \quad \forall (u,v)\in E , \, \forall e \in E  \label{ilp:dinitz2011dual_4}\\
\end{align}

The next step is to look at the constraint of the dual problem and observe whether checking a subset of constraints can be done in polynomial time (in more technical terms we need to construct a separation oracle for the dual LP; see \cite{bertsimas1997introduction}). In this case, the only nontrivial set of constraints is of the form  $z_{u,v} - \sum_{e\in E} y_{u,v}^e \leq 0$, and we have an exponential number of constraints of this type. One then shows that a polynomial subset of these constraints gives a good approximation to the dual problem. More specifically, the rest of the constraints are violated only by a factor of $(1-\epsilon)$, i.e. there might be paths with $(1-\epsilon)z_{u,v} \leq \sum_{e\in E} y_{u,v}^e $.

We can, therefore, obtain an optimal solution to the dual problem up to a $(1-\epsilon)$ factor. Using a strong version of the duality theorem, if we remove the variables from the primal problem corresponding to these constraints, we acquire the following result. 

\begin{theorem}[\cite{dinitz2011directed}]
For any $\epsilon \geq 0$, there is an algorithm that in polynomial time computes a $(1+\epsilon)$--approximation to the optimal solution of above primal LP problem \eqref{ilp:dinitz2011primal}--\eqref{ilp:dinitz2011primal_5}.
\end{theorem}

Dinitz and Krauthgamer also prove an integrality gap for this problem:

\begin{theorem}[\cite{dinitz2011directed}]\label{THM:DinitzIntegralityGap}
The integrality gap of the primal LP~\eqref{ilp:dinitz2011primal}--\eqref{ilp:dinitz2011primal_5} for the directed $t$--spanner problem is $\Omega(\frac{1}{t} n^{\frac{1}{3} - \epsilon})$, for any constant $\epsilon>0$.
\end{theorem}
}

\subsection{ILPs for Pairwise Spanners}

Now we turn to considering ILPs which exactly solve the spanner problem for a given input graph. ILP solvers typically require significantly more computation than LP relaxations, but are practical in some instances.

Let $G=(V,E)$ be a graph (weighted or unweighted), let $P\subseteq V\times V$ be the desired pairs for which the spanner condition will hold, and let $t\geq1$ be a prescribed stretch factor. Recall that $w_e$ is the weight of an edge $e\in E$ (with $w_e = 1$ in the unweighted case).  Let $x_e = 1$ if $e$ is selected in the spanner, and $x_e=0$ otherwise. Further, let $\mathcal{P}_{uv}$ denote the set of all $u$--$v$ paths of length at most $t \cdot d_G(u,v)$, let $\mathcal{P}$ denote the union of all such paths ($\mathcal{P} = \bigcup_{(u,v) \in P} \mathcal{P}_{uv}$), and let $y_\rho = 1$ if the path $\rho \in \mathcal{P}$ is selected in the spanner, and 0 otherwise. Let $\delta_\rho^e = 1$ if edge $e$ is on path $\rho$, and 0 otherwise. The ILP formulation by Sigurd and Zachariasen for the pairwise multiplicative spanner problem is as follows:

\begin{align}
    \allowdisplaybreaks
    \text{min} \sum_{e \in E} w_ex_e \text{ subject to}\label{eqn:ilp-sigurd-obj}\\
    \sum_{\rho \in \mathcal{P}_{uv}} y_\rho \delta_\rho^e &\le x_e & \forall e \in E; \forall (u,v) \in P \label{eqn:ilp-sigurd-1}\\
    \sum_{\rho \in \PP_{uv}} y_\rho &\ge 1 & \forall (u,v) \in P \label{eqn:ilp-sigurd-2}\\
    x_e &\in \{0,1\} & \forall e \in E \label{eqn:ilp-sigurd-3} \\
    y_\rho &\in \{0,1\} & \forall \rho \in \PP \label{eqn:ilp-sigurd-4}
\end{align}

Constraint \eqref{eqn:ilp-sigurd-1} enforces that if a path $\rho \in \PP_{uv}$ is selected in the spanner, then all edges on that path are selected as well. Constraint \eqref{eqn:ilp-sigurd-2} enforces that every pair $(u,v) \in P$ has at least one selected $t$--spanner path.  


This ILP has $O(|\PP|)$ constraints which can be exponential in the size of the graph; however, the authors apply column generation to solve a restricted master problem (i.e. using only a subset of $\PP$ for the constraints), and show that this procedure yields an optimal solution.  The authors then use the ILP to show that the greedy algorithm of Alth\"{o}fer et al.~(see Section \ref{SEC:Greedy}) performs optimally in many cases.

{\color{black}Sigurd and Zachariasen left as an open question whether the pairwise spanner problem admits a practical polynomial size ILP. Recently, Ahmed et al.~\cite{MLGS_proceeding,MLGS2} proposed an affirmative resolution to this question, giving a polynomial-size flow-based ILP formulation for the pairwise spanner problem which can be used to compute minimum-weight graph spanners with arbitrary distortion function, even if the input graph is directed, and experiments were provided in support of its feasibility.}

The ILP of Ahmed et al.~\cite{MLGS_proceeding,MLGS2} is set up as follows.  Again let $P\subseteq V\times V$, and let $f:\R_+\to\R_+$ be a distortion function satisfying $f(x)\geq x$ for all $x$ (note the function need not be continuous). Replacing each edge $(u,v)$ with two directed edges $(u,v)$ and $(v,u)$ of equal weight. Let $\bar{E}$ be the set of all directed edges (so $|\bar{E}| = 2|E|$). Given $(i,j) \in \bar{E}$ and $(u,v) \in P$, the variable $x_{(i,j)}^{uv}$ is $1$ if edge $(i,j)$ is included in the selected $u$--$v$ path in the spanner $G'$, or 0 otherwise.

The formulation is then as follows; by $\In(v), \Out(v)$ we mean the set of incoming and outgoing edges to $v$ (respectively), and the vertices are arbitrarily ordered so that constraints (\ref{eqn:ilp-2})--(\ref{eqn:ilp-4}) are well-defined.

\begin{align}
    \allowdisplaybreaks
    \text{Minimize} \sum_{e \in E} w_ex_e \text{ subject to} \label{eqn:ilp-obj}\\
    \sum_{(i,j) \in \bar{E}} x_{(i,j)}^{uv} w_e &\le f(d_G(u,v)) & \hspace{-10pt}\forall (u,v) \in P, u<v; e = (i,j)\in E \label{eqn:ilp-1}\\
    \sum_{(i,j) \in \Out(i)} x_{(i,j)}^{uv} - \sum_{(j,i) \in \In(i)} x_{(j,i)}^{uv} &= \begin{cases}
        1 & i = u \\
        -1 & i = v \\
        0 & \text{else}
    \end{cases} & \hspace{-10pt} \forall (u,v) \in P, u<v; \forall i \in V \label{eqn:ilp-2} \\
    \sum_{(i,j) \in \Out(i)} x_{(i,j)}^{uv} & \le 1 &\forall (u,v) \in P, u<v; \forall i \in V \label{eqn:ilp-3}\\
    x_{(i,j)}^{uv} + x_{(j,i)}^{uv} &\le x_e & \hspace{-30pt}\forall (u,v) \in P, u<v; \forall e = (i,j) \in E \label{eqn:ilp-4}\\
    x_e, x_{(i,j)}^{uv} &\in \{0,1\} \label{eqn:ilp-5}
\end{align}
This ILP has $2|E||P|$ binary variables, or $2|E|\binom{|V|}{2} = |E||V|(|V|-1)$ variables in the full spanner problem where $P=V\times V$.
In the multiplicative spanner case, the number of ILP variables can be significantly reduced (see \cite{MLGS_proceeding} for more details).  These reductions are somewhat specific to multiplicative spanners, and so it would be interesting to determine if other simplifications are possible for more general distortion.
We note that the distortion $f$ can have any form, which allows for considerable flexibility in the types of pairwise spanners represented by the above ILP.


{\color{black}
Dinitz, Nazari and Zhang \cite{dinitz2019lasserre} prove a stronger integrality gap than Theorem \ref{THM:DinitzIntegralityGap} by considering Lasserre lifts of the flow LP. They essentially show that even the strongest lift and project methods cannot help significantly in the approximation of spanners.

Chlamt\'{a}\v{c}, Dinitz, and Krauthgamer  \cite{chlamtac2012} give an LP relaxation for the lowest-degree 2--spanner (LD--2SP) problem using the Sherali--Adams lift method. They establish a polynomial time algorithm with approximation ratio $O(\Delta^{3-2\sqrt{2}})$, where $\Delta$ is the maximum degree of the graph.

Chlamt\'{a}\v{c} and Dinitz \cite{chlamtac2016} formulate an (I)LP for the lowest-degree $k$--spanner (LD--$k$SP) problem. They also provide a hardness of approximation result based on the Min-REP problem. The basic idea is that small REP-covers correspond to small spanners, when constructing a yes/no decision oracle. We refer the reader to \cite{chlamtac2016} for the full details of REP.

\begin{theorem}
For any integer $k \geq 3$, there is no polynomial-time algorithm that can approximate
lowest degree $k$-spanner better than  $\Delta^{\Omega(1/k)}$ unless $NP \subset BPTIME(2^{polylog(n)})$, where $\Delta$ is the maximum degree of the input graph.
\end{theorem}

Berman et al.~\cite{BERMAN201393} formulate an instance of anti-spanner LP problems. An anti-spanner is a subset of edges that don't form a $t$--spanner and is maximal with respect to this property.

Dinitz and Zhang \cite{dinitz2016approximating} prove that the approximation ratio for the $t$--spanner problem is at most $O(n^{\frac{1}{3}})$. They also prove an integrality gap for a flow-based LP. 
}

\section{Distributed and Streaming Algorithms}\label{SEC:Distributed}
\begin{itemize}
\setlength{\itemindent}{.2in}
\item[\cite{elkin2005computing}] \bibentry{elkin2005computing}
{\color{black}\item[\cite{elkin2006efficient}] \bibentry{elkin2006efficient}}
\item[\cite{derbel2007deterministic}] \bibentry{derbel2007deterministic}
\item[\cite{derbel2008locality}] \bibentry{derbel2008locality}
\item[\cite{baswana2008streaming}] \bibentry{baswana2008streaming}
\item[\cite{pettie2008distributed}] \bibentry{pettie2008distributed}
{\color{black} \item[\cite{elkin2011streaming}] \bibentry{elkin2011streaming} }
\item[\cite{lenzen2013efficient}] \bibentry{lenzen2013efficient}
\item[\cite{kapralov2014spanners}] \bibentry{kapralov2014spanners}
\item[\cite{censor2016distributed}] \bibentry{censor2016distributed}
\item[\cite{elkin2016fast}] \bibentry{elkin2016fast}
\item[\cite{chechik2018near}] \bibentry{chechik2018near}
{\color{black} \item [\cite{elkin2018efficient}] \bibentry{elkin2018efficient}}
\item[\cite{censor2018sparsest}] \bibentry{censor2018sparsest}
\item[\cite{elkin2019near}] \bibentry{elkin2019near}
{\color{black}\item[\cite{elkin2019distributed}] \bibentry{elkin2019distributed}}
\end{itemize}

{\color{black}In distributed algorithms, computations happen at each vertex of the graph and processors communicate to collectively accomplish a graph-level task. In Peleg's terminology \cite{Peleg2000Distributed}, two models of distributed computing are LOCAL and CONGEST, where respectively, unbounded and short messages are communicated between processors.
Distributed algorithms for computing spanners are considered particularly important because a common use of spanners is to preprocess a distributed network, sparsifying the network without changing its communication latency too much.
Indeed, the historically first appearances of spanners were in applications to \emph{synchronizers} \cite{awerbuch1985complexity, peleg1989optimal, peleg1989graph}, which are protocols that convert between the asynchronous distributed model (where nodes send messages to their neighbors at unpredictable times) and the synchronous distributed model (where nodes send messages at each tick of a centralized clock).
It was only later that applications to centralized computation were developed, and that spanners were considered an interesting object of study in their own right.}

Elkin~\cite{elkin2005computing} gives a randomized distributed algorithm which produces a $(1+\eps,\beta)$--spanner with $\beta = \left(\frac{k}{\eps}\right)^{O(\log k)}\rho^{-\frac1\rho}$ of size $\widetilde{O}(\beta n^{1+\frac1k})$ in time $O(n^{1+\frac{1}{2k}})$.  Elkin and Matar~\cite{elkin2019near} study distributed algorithms computed in the CONGEST model for finding near-additive spanners, i.e. $(1+\eps,\beta)$--spanners.  The difference of their approach from previous works is that the algorithm is deterministic.  The spanners constructed have three parameters $\eps, k,\rho$ and require $\beta = \left(\frac{O(\log k\rho+\rho^{-1})}{\rho\eps}\right)^{\log k\rho+\rho^{-1}+O(1)}$, and yield $(1+\eps,\beta)$--spanners of size $O(\beta n^{1+\frac1k})$ with time complexity $O(\beta n^\rho \rho^{-1})$.

Pettie~\cite{pettie2008distributed} describes algorithms for computing sparse low distortion spanners in distributed networks and provides some non-trivial lower bounds on the tradeoff between time, sparseness, and distortion. The algorithms assume a synchronized distributed network, where relatively short messages may be communicated in each time step. The first result is a fast distributed algorithm for finding an $O(2^{\log^* n} \log n)$--spanner with size $O(n)$.\footnote{Recall that $\log^* n$ is the inverse tower function; essentially, it is the least height of a tower $2^{2^{2^{\dots}}}$ whose value is $>n$.}
The second result is a new class of efficiently constructible $(\alpha, \beta)$--spanners called Fibonacci spanners whose distortion improves with the distance being approximated. At their sparsest, Fibonacci spanners can have nearly linear size, namely {\color{black} $O(n(\log \log n)^\phi)$},
where $\phi = \frac{1 + \sqrt{5}}{2}$ is the golden ratio. 

Lenzen and Peleg \cite{lenzen2013efficient} show that additive $2$--spanners can be computed in the CONGEST model in $O(n^\frac12\log n+\textnormal{diam}(G))$ rounds.  In \cite{censor2018sparsest}, the authors give a sequential algorithm for computing an additive $6$--spanner for unweighted graphs which yields near-optimal sparsity $O(n^\frac43\log^\frac43n)$ edges, but allows for a distributed construction algorithm (using the CONGEST model of computation) which runs via an efficient construction of weighted BFS trees in $O(n^\frac23 \log^{-\frac13}n +\textnormal{diam}(G))$ rounds. 

In Derbel et al.~\cite{derbel2007deterministic} several deterministic distributed algorithms have been provided to compute different kinds of spanners. They give an algorithm to construct an $O(k)$--spanner of an unweighted graph with $O(kn^{1+\frac{1}{k}})$ edges in $O(\log^{k-1} n)$ time.  Also, they have shown that in $O(\frac{\log n}{\epsilon})$ time one can construct a spanner with $O(n^\frac{3}{2})$ edges that is both a $3$--spanner and a $(1 + \epsilon, 8 \log n)$--spanner. Furthermore, they have shown that in $n^{O\left(\frac{1}{\sqrt{\log n}}\right)} + O\left(\frac{1}{\epsilon}\right)$ time one can construct a spanner with $O(n^\frac{3}{2})$ edges which is both a $3$--spanner and a $(1+\eps, 4)$--spanner.  In \cite{derbel2008locality}, the authors propose a deterministic, distributed algorithm that computes a $(2k-1,0)$--spanner in $k$ rounds of computation (or $3k-2$ rounds if $n$ is not known) with $O(kn^{1+\frac{1}{k}})$ edges.
Also, it is further shown that no (randomized) algorithm producing $O(n^{1+\frac{1}{k}})$ size spanners (possibly with additive stretch) can run distributively in sub-polynomial rounds of computation.
{\color{black} Recently, Elkin, Filtser, and Neiman \cite{elkin2019distributed} gave distributed constructions of spanners with good controls on lightness.}

Censor-Hillel et al.~\cite{censor2016distributed} study the complexity of distributed constructions of purely additive spanners. The provided spanner algorithms have three general steps: first, each node tosses a coin to be a cluster center; second, each cluster center tosses another coin to be a BFS tree; third, add to the current graph edges that are part of certain short paths.

{\color{black}
Elkin \cite{elkin2011streaming} gave a streaming algorithm for building multiplicative spanners with near-optimal size/stretch tradeoff and $O(1)$ expected processing time per edge.
A close variant of this algorithm also gives a fully dynamic algorithm, with similar update time, for maintaining spanners with near-optimal size/stretch tradeoff.
(A similar result was proved simultaneously by Baswana~\cite{baswana2008streaming}.)
Streaming algorithms for mixed spanners were given in \cite{elkin2006efficient, elkin2018efficient}.}

Feigenbaum et al.~\cite{feigenbaum2005graph} consider the semi-streaming model where at most a logarithmic number of passes over the graph stream are allowed. It was shown that a single pass is enough to compute a $(1+\eps)\log n$--spanner for a weighted undirected graph under mild conditions for weight values.

{\color{black}A somewhat distinct, yet relevant, body of works pertains to dynamic stream models where both insertion and deletion are allowed~\cite{ahn2012graph}. Kapralov and Woodruff ~\cite{kapralov2014spanners} construct linear sketches of graphs that (approximately) preserve the spectral information of the graph in a few passes over the stream of graph data. They use a multi-layer clustering approach to construct a multiplicative $2^k$--spanner using $O(n^{1+\frac{1}{k}})$ bits of space from a dynamic stream of inputs.}


\section{Special Types of Spanners}\label{SEC:TypesOfSpanners}

In this section, we describe some of the many variants of spanners that have been considered in the literature. We emphasize primarily the definitions and basic results for brevity, and readers are encouraged to consult the listed papers for more information.

\subsection{\label{SEC:emulators}Emulators}
\begin{itemize}
\setlength{\itemindent}{.2in}
\item[\cite{dor2000all}] \bibentry{dor2000all}
\item[\cite{thorup2006spanners}] \bibentry{thorup2006spanners}
\item[\cite{Hsien2018nearoptimal}] \bibentry{Hsien2018nearoptimal}
\end{itemize}

The word \emph{spanner} generally implies a \emph{subgraph} that approximates the distance metric of an input graph.
When the approximating graph can instead be arbitrary (not necessarily a subgraph), it is called an \emph{emulator}.
Spanners are a special case of emulators.
Emulators are typically allowed to be weighted and sometimes directed, even if the input graph is not.

Emulators hold an interesting place in the literature because optimal bounds are essentially known for the size/error tradeoff of emulators, whereas the optimal bounds for spanners remain open.
The first example of this is for emulators with purely additive error.
In \cite{dor2000all}, the authors prove that every $n$--node undirected unweighted input graph has an emulator on $O(n^\frac{4}{3})$ edges with $4$--additive error.
This is optimal in the sense that neither the edge bound nor the additive error in this result can be unilaterally improved (this follows from the Girth Conjecture which is confirmed in the relevant setting $k=3$; see Section \ref{SEC:girth}).
However, it is only known that every graph has a \emph{spanner} on $O(n^\frac{4}{3})$ edges and $6$--additive error \cite{baswana2010additive}, and it is still open whether one can unilaterally improve the additive error to $4$.

Essentially tight bounds for emulators with sublinear error are known as well.
An important result in \cite{thorup2006spanners} is a simple construction of emulators with $O(n^{1 + \frac{1}{2^k - 1}})$ edges and error function $f(d) = d + O(d^{1 - \frac{1}{k}})$, for any fixed positive integer $k$.
The main result of \cite{abboud2018hierarchy} is that this is essentially optimal; 
{\color{black}see Section \ref{SEC:LowerBounds} for more details.}

Emulators have also been central to a line of research seeking fast approximation algorithms for all-pairs shortest paths (APSP).
In \cite{dor2000all}, the authors give approximate APSP algorithms essentially by constructing a series of emulators and then running standard shortest path algorithms on these emulators, which are fast due to their sparsity.
A related approach was taken in \cite{Hsien2018nearoptimal}: the authors prove that, for any undirected unweighted planar input graph $G = (V, E)$ and set of terminals $T \subseteq V$, there is an emulator on $\widetilde{O}(\min\{|T|^2, \sqrt{n|T|}\})$ that \emph{exactly} preserves the distances between node pairs in $T$ (called a \emph{distance emulator} in analogy with distance preservers).
With a similar technique to the above, the authors then show that one can compute distances between any $|T| = \widetilde{O}(n^\frac{1}{3})$ terminals in an $n$--node planar graph in $\widetilde{O}(n)$ time.

\subsection{Approximate Distance Oracles}
\label{sec:oracles}
\begin{itemize}
\setlength{\itemindent}{.2in}
\item[\cite{thorup2001oracle}] \bibentry{thorup2001oracle}
\item[\cite{thorup2005det}] \bibentry{thorup2005det}
\item[\cite{baswana2006approximate}] \bibentry{baswana2006approximate}
\item[\cite{baswana2008distance}] \bibentry{baswana2008distance}
\item[\cite{sommer2014shortest}] \bibentry{sommer2014shortest}
\end{itemize}

Related to emulators are \textit{approximate distance oracles}.  Given a graph $G$, a {\it distance oracle} is a data structure that can answer exact distance queries between any pair of vertices $u,v \in V$. By preprocessing the graph using algorithms such as All-Pairs Shortest Path (APSP), these queries can be answered in constant time. However, algorithms like APSP are inefficient for this purpose as they have a runtime of $O(n^3)$ and space requirement of $O(n^2)$. If approximate distances are acceptable to the user, then the runtime and space bounds can be much improved. An \textit{approximate distance oracle} is a data structure which approximately answers distance queries, and similar to a spanner, the answer of a query is often allowed to be at most $t$ times the actual distance between the two vertices. In this case, spanners are special cases of approximate distance oracles, but are more restrictive given that they need to be subgraphs.  Related objects are \textit{emulators} which are graphs $G^*=(V,E^*)$ on the same vertices as the original graph, but with different edge sets, but so that distances in $G^*$ approximate distance in $G$~\cite{thorup2006spanners}.  

Thorup and Zwick~\cite{thorup2001oracle} propose a randomized algorithm 
to compute an approximate distance oracle with a stretch $(2k-1)$ such that for all $u,v \in V,$ the oracle returns a distance $d_k(u,v)$ with $d_k(u,v) \le (2k-1)d_G(u,v)$.  Note that the algorithm will not return a subgraph of $G$, and hence is not a spanner algorithm, but is rather returning an approximation to the length of the shortest paths in $G$. This algorithm has an improved $O(kmn^\frac{1}{k})$ runtime and $O(kn^{1+\frac{1}{k}})$ space requirement. 

As a follow-up to~\cite{thorup2001oracle}, Roditty et al.~\cite{thorup2005det} propose two extensions: restricting the approximation to a set of sources $S \subseteq V$, and derandomizing the algorithm. 
The former results in an algorithm with $O(km|S|^\frac{1}{k})$ runtime and $O(kn|S|^\frac{1}{k})$ space requirement, which becomes very useful when $|S|\ll n$. 
The deterministic variant of the algorithm leads to an increase in runtime by a logarithmic factor but retains the same space requirements as the randomized algorithm.

Finally, Roditty et al.~\cite{thorup2005det} extend the derandomization technique to the linear time algorithm by Baswana and Sen~\cite{Surender03ALP} to compute a $(2k-1)$--spanner for a given graph $G$. The randomization in this algorithm involves the initial step of choosing cluster centers randomly with probability $|V|^{-\frac{1}{2}}$. The authors propose doing this deterministically by building a bipartite graph of $B=(V_1,V_2,E)$ and applying the linear time algorithm for computing a closed dominating set on this graph to find the cluster centers.

Baswana and Sen~\cite{baswana2006approximate} provide an algorithm to construct approximate distance oracles in expected $O(n^2)$ time if the graph is unweighted. One of the new ideas used in the algorithm also leads to the first expected linear time algorithm for computing an optimal size $(2,1)$--spanner of an unweighted graph, whereas existing algorithms had $\Theta(n^2)$ running time. In \cite{baswana2008distance}, they break this quadratic barrier at the expense of introducing a (small) constant additive error for unweighted graphs. In achieving this goal, they have been able to preserve the optimal size--stretch tradeoffs of the oracles. One of their algorithms can be extended to weighted graphs, where the additive error becomes $2w_{max}(u, v)$, where $w_{max}(u, v)$ is the heaviest edge in the shortest path between vertices $u$ and $v$.

In~\cite{dor2000all}, a $\widetilde{O}(\min\{n^\frac{3}{2}m^\frac{1}{2}, n^\frac{7}{3}\})$--time algorithm for computing all distances in a graph with an additive error of at most 2 has been described. For every even $k > 2$, they describe an $\widetilde{O}(\min\{$ $n^{2-\frac{2}{(k+2)}}m^{\frac{2}{(k+2)}}, n^{2+\frac{2}{(3k-2)}}\})$--time algorithm for computing all distances with an additive one-sided error of at most $k$. Then, they show that every unweighted graph on $n$ vertices has an additive $2$--spanner with $\widetilde{O}(n^\frac{3}{2})$ edges and a $4$--spanner with $\widetilde{O}(n^\frac{4}{3})$ edges. Finally, they show that any weighted undirected graph on $n$ vertices has a multiplicative $3$--spanner with $\widetilde{O}(n^\frac{3}{2})$ edges and that such a $3$--spanner can be built in $\widetilde{O}(mn^\frac{1}{2})$ time. 

Sommer~\cite{sommer2014shortest} gives a survey on the shortest path query problem, with both theoretical results and practical applications, and should be consulted for further reading. This problem has numerous applications including Geographic Information Systems (GIS), proximity search in a database, packet routing, and metabolic networks. The shortest path query problem is different than the classical single source shortest path problem and APSP problem in that there are two steps: the first is a preprocessing algorithm, and the second is to process queries efficiently. The preprocessing algorithm provides a data structure which can be used to process multiple queries without changing the structure. Unlike the APSP problem, here the data structure is computed in a space efficient manner.

\subsection{Diameter and Eccentricity Spanners}
\begin{itemize}
\setlength{\itemindent}{.2in}
\item[\cite{backurs2018diameter}] \bibentry{backurs2018diameter}


\item[\cite{choudhary2020extremal}]{\color{black}\bibentry{choudhary2020extremal}}
\end{itemize}

Choudhary and Gold~\cite{choudhary2020extremal} consider the following problem: given a directed graph $G = (V,E)$, compute a subgraph $H$ with the property that the graph diameter is preserved up to a multiplicative factor $t$, which the authors term a $t$--\emph{diameter spanner}. The authors show that, given an unweighted directed graph $G$, there is a polynomial time algorithm that computes a 1.5--diameter spanner containing {\color{black} $O(n^\frac{3}{2} \sqrt{\log n})$ }edges, in time $\widetilde{O}(m\sqrt{n})$ deterministically. They define the $t$-\emph{eccentricity spanner} similarly, where the eccentricity\footnote{The \emph{eccentricity} of a vertex $v \in V$ is the maximum graph distance between $v$ and any other vertex.} of each vertex $v$ in $H$ should be no more than $t$ times the eccentricity of $v$ in $G$. The authors show that given a weighted directed graph, there is a Las Vegas algorithm which computes a 2--eccentricity spanner containing $O(n \log^2 n)$ edges in $\widetilde{O}(m)$ expected time.

\subsubsection*{Hardness}
{\color{black}Backurs et al.~\cite{backurs2018diameter} show that unless the Strong Exponential Time Hypothesis\footnote{Let $s_k := \inf\{\delta: \text{$k$--SAT has an } O(2^{n\delta})\text{--time algorithm}\}$. The Strong Exponential Time Hypothesis (SETH) asserts that $\lim_{k \to \infty} s_k = 1$.} (SETH) fails, there is no $O(n^{\frac{3}{2} - \delta})$ time algorithm ($\delta > 0$) which can approximate the diameter of a sparse weighted graph with ratio better than $\frac{5}{3}$. Unless SETH fails, no $O(n^{\frac32 - \delta})$--time algorithm can approximate the eccentricities of an undirected, unweighted sparse graph with ratio better than $\frac{9}{5}$.}
\subsection{Ramsey Spanning Trees}

\begin{itemize}
\setlength{\itemindent}{.2in}
\item[\cite{abraham2018ramsey}] \bibentry{abraham2018ramsey}
\end{itemize}

Abraham et al.~\cite{abraham2018ramsey} consider a generalization of the metric Ramsey problem to graphs.  The problem is to find a subset $S\subseteq V$ of sources and a spanning tree $T$ which gives a $t$--spanner for small $t$; i.e. one wants to simultaneously find the sources and sourcewise spanner for a given graph.  Their main result is to find (in polynomial time) a subset $S$ of size at least $n^{1-\frac1k}$ and spanning tree $T$ which is an $(O(k\log\log n),0,S\times V)$--spanner of $G$.

\subsection{Slack Spanners}
\begin{itemize}
\setlength{\itemindent}{.2in}
    \item[\cite{chan2006spanners}] \bibentry{chan2006spanners}
\end{itemize}
In~\cite{chan2006spanners}, the authors consider the problem of computing a subgraph $H$ that preserves shortest distances to all but an $\epsilon$ fraction of vertices. Specifically, $H = (V,E_H)$ is an $\epsilon$--\emph{slack spanner} with distortion $D$ if for all $v \in V$, the distances in $H$ are preserved for the furthest $(1-\epsilon)n$ vertices from $v$ up to distortion $D$. The authors show that given a graph $G$, one can find a subgraph $H$ which is a $O(\log \frac{1}{\epsilon})$--distortion $\epsilon$--slack spanner containing $O(n)$ edges, for every $\epsilon$. {\color{black}Furthermore, they showed that gracefully degrading existing spanners also have $O(1)$ average stretch with size $O(n)$.}

\subsection{Spanners in Sparse Graphs}
\begin{itemize}
\setlength{\itemindent}{.2in}
\item[\cite{dragan2011spanners}] \bibentry{dragan2011spanners}
\end{itemize}

Dragan et al.~\cite{dragan2011spanners} show that the problem of computing a $t$--spanner of a planar graph $G$ of treewidth\footnote{The \emph{treewidth} of a graph $G$ is the minimum width over all tree decompositions of $G$, where the width of a tree decomposition equals 1 less than the maximum size of a vertex set in the decomposition.} at most $k$ is fixed parameter tractable parametrized by $k$ and $t$. They then show that the problem is also fixed parameter tractable on graphs that do not contain a fixed apex graph\footnote{An \emph{apex graph} is a graph that is planar if one vertex is removed.} as a minor. 




\subsection{Steiner $t$-Spanners}
\begin{itemize}
\setlength{\itemindent}{.2in}
\item[\cite{althofer1993sparse}] \bibentry{althofer1993sparse}
\item[\cite{handke2000tree}] \bibentry{handke2000tree}

\end{itemize}

Given an unweighted graph $G = (R \cup S, E)$ where $R$ is the set of terminals and the induced graph of $R$ is connected, a subgraph $H$ of $G$ is a \emph{Steiner $t$--spanner} of $G$ if $d_H(u,v) \leq t d_G(u,v),$ for all $ u,v \in R$. Alth\"{o}fer et al.~\cite{althofer1993sparse} deal with absolute lower bounds on the number of edges that an Steiner $t$--spanner can have.

\subsubsection*{Hardness}
{\color{black}Handke and Kortsarz~\cite{handke2000tree} show that the minimum Steiner $t$--spanner cannot be approximated with ratio $(1-\varepsilon)\log |S|$ for any $\varepsilon > 0$, unless $\text{NP} \subseteq \text{DTIME}(n^{\log \log n})$ where $S = V \setminus R$ is the set of Steiner vertices.}

\subsection{Tree $t$--spanners} \label{SEC:tree-t-spanner}

\begin{itemize}
\setlength{\itemindent}{.2in}
\item[\cite{cai1995tree}] \bibentry{cai1995tree}
\item[\cite{emek2008approximating}] \bibentry{emek2008approximating}
    \item [\cite{alvarez2018mixed}] \bibentry{alvarez2018mixed}
\end{itemize}
A \emph{tree $t$-spanner} of a graph $G$ is a spanning tree of $G$ which is also a $t$-spanner.  Cai and Corneil~\cite{cai1995tree} show that a tree 1-spanner, if it exists, can be found in $O(m \log \beta(m,n))$ time, where $\beta(m,n) = \min\{i : \log^i n \le m/n\}$. Further, a tree 2-spanner (if it exists) can be found in linear $O(m+n)$ time. For $t \ge 4$, determining whether a tree $t$-spanner exists is NP-hard via a reduction from 3-SAT.

Emek and Peleg~\cite{emek2008approximating} consider the optimization version, referred to as \emph{minimum max-stretch spanning tree}: given a graph $G$, find a spanning tree minimizing the stretch factor $t$. Emek and Peleg give a divide-and-conquer $O(\log n)$ approximation algorithm for MMST.  \'{A}lvarez-Miranda and Sinnl~\cite{alvarez2018mixed} give a compact, polynomial-size mixed integer program (MIP) for this problem. They obtain a formulation with fewer variables by using a cut-based formulation to impose a spanning tree topology, then using Benders decomposition to get rid of path variables.  For other approximation algorithms, see~\cite{singh2018artifical,sundar2019steady}.

\subsubsection*{Hardness}
{\color{black}Emek and Peleg~\cite{emek2008approximating} show that it is NP--hard to approximate the MMST problem with ratio better than $\frac{5}{4}$, and that MMST cannot be approximated additively by any $o(n)$ term unless P = NP.}

\subsection{Lowest-degree $k$-Spanners} \label{SEC:LDkSP}
\begin{itemize}
\setlength{\itemindent}{.2in}
    \item[\cite{Kortsarz1994GeneratingL2}] \bibentry{Kortsarz1994GeneratingL2}
    \item[\cite{chlamtac2012}] \bibentry{chlamtac2012}
    \item[\cite{chlamtac2016}] \bibentry{chlamtac2016}
\end{itemize}
Given integer $k \ge 2$, the \emph{lowest-degree $k$-spanner} (LD-$k$SP) problem is to compute a $k$-spanner of a given undirected graph $G$, with the objective of minimizing the maximum degree of any vertex in the spanner. Let $\Delta$ denote the maximum vertex degree of the input graph $G$.

For $k=2$, Kortsarz and Peleg~\cite{Kortsarz1994GeneratingL2} give an LP-based $\tilde{O}(\Delta^{1/4})$-approximation for LD-2SP. This was improved more recently by Chlamt\'{a}\v{c} et al.~\cite{chlamtac2012}, who give an $\tilde{O}(\Delta^{3-2\sqrt{2}-\varepsilon}) \approx \tilde{O}(\Delta^{0.172})$-approximation for LD-2SP. For $k \ge 3$, Chlamt\'{a}\v{c} et al.~\cite{chlamtac2016} give an LP-based $\tilde{O}(\Delta^{(1-(1/k))^2})$-approximation for LD-$k$SP (see Section~\ref{SEC:lp-more}).

\subsubsection*{Hardness}
{\color{black}For $k=2$, Kortsarz and Peleg show that LD--2SP cannot be approximated with ratio better than $(\ln n)/5$ unless $\text{NP} \subseteq \text{DTIME}(n^{\log \log n})$ via a reduction from the set cover\footnote{The \emph{set cover} problem is as follows: given a universe $U=\{1,2,\ldots,n\}$, a collection of sets $S_1, \ldots, S_m$ such that $\bigcup_{i=1}^m S_i = U$, and an integer $t \ge 0$, determine if there is a collection of $t$ sets whose union is $U$.} problem.

For $k \ge 3$, Chlamt\'{a}\v{c} and Dinitz~\cite{chlamtac2016} show that LD--$k$SP admits no polynomial time approximation with ratio $\Delta^{\Omega\left(\frac{1}{k}\right)}$ unless $\text{NP} \subseteq \text{BPTIME}(2^{\text{polylog}(n)})$. Further, if the input graph is directed (in which case the ``degree'' of a spanner vertex is the sum of its in-degree and out-degree), then directed LD--$k$SP cannot be approximated with ratio $\Delta^{O(1)}$ unless $\text{NP} \subseteq \text{DTIME}(2^{\text{polylog}(n)})$.}

\subsection{{\color{black} Spanners for Doubling Metrics}}
\begin{itemize}
\setlength{\itemindent}{.2in}
\item[\cite{arya1995euclidean}] \bibentry{arya1995euclidean}
\item[\cite{elkin2015optimal}] \bibentry{elkin2015optimal}
\item[\cite{chan2015new}] \bibentry{chan2015new}
\item[\cite{borradaile2019greedy}] \bibentry{borradaile2019greedy}
\item[\cite{Gottlieb2015}] \bibentry{Gottlieb2015}
\end{itemize}

{\color{black}
For a metric space $M = (V, d)$, its \emph{doubling dimension} $\mathcal{D}(M)$ is the quantity 
\begin{align*}
\mathcal{D}(M) := \inf \{d \in \mathbb{R} \ : \ &\text{for all } v \in V, r>0, \text{ the ball } B(v, r) \text{ can be covered by }\\
&2^d \text{ balls of radius } r/2.\}
\end{align*}

There is a body of work that builds spanners for graphs assuming bounds on the doubling dimension of their associated metric space.
This is motivated in part as a natural stepping stone between general graph spanners and \emph{geometric spanners}, where the input graphs are essentially finite submetric spaces of Euclidean space (the book \cite{Narasimhan:2007:GSN:1208237} covers these extensively, but they are out of scope for this survey).
In particular, a basic fact is that $\mathbb{R}^d$ and any submetric space thereof has doubling dimension $\Theta(d)$, and some results for geometric spanners in $\mathbb{R}^d$ generalize readily to any graph with similar doubling dimension.
For example, there was a famous conjecture by Arya et al.~\cite{arya1995euclidean} that, given an $n$-point submetric space of $\mathbb{R}^{O(1)}$, in $O(n \log n)$ time one can construct a geometric spanner with constant maximum degree, $O(\log n)$ hop-diameter (i.e., maximum number of edges in a shortest path), and $O(\log n)$ lightness (i.e., total weight $O(\log n) \cdot W(\MST(G))$).
This conjecture was proved by Elkin and Solomon~\cite{elkin2015optimal} (and later simplified by Chan et al.~\cite{chan2015new}), more generally in any metric space of constant doubling dimension.

An important fact about graphs of low doubling dimension is that the greedy construction algorithm (presented in Section \ref{SEC:Greedy}) has an improved lightness analysis: following \cite{Gottlieb2015}, Borradaile et al.~\cite{borradaile2019greedy} proved that a greedy $(1+\eps)$--spanner of a graph of doubling dimension $d$ has total weight $\eps^{-O(d)} W(\MST(G))$.
(The corresponding bound for geometric graphs, which is tight, is in \cite{Narasimhan:2007:GSN:1208237}.)
}


\section{Spanners For Changing Graphs}\label{SEC:ChangingGraphs}

In recent years, there have been variants of spanners studied in the literature which typically involve changes allowed in the initial graph $G$, for example edge addition or deletion.

\subsection{Fault Tolerant Spanners}

\begin{itemize}
\setlength{\itemindent}{.2in}
\item[\cite{chechik2010fault}] \bibentry{chechik2010fault}
\item[\cite{ausiello2010computing}] \bibentry{ausiello2010computing}
\item[\cite{dinitz2011fault}] \bibentry{dinitz2011fault}
\item[\cite{braunschvig2015fault}] \bibentry{braunschvig2015fault}
\item[\cite{bilo2015improved}] \bibentry{bilo2015improved}
\item[\cite{parter2017fault}] \bibentry{parter2017fault}
\item[\cite{bodwin2018optimal}] \bibentry{bodwin2018optimal}
\item[\cite{bodwin2018trivial}] \bibentry{bodwin2018trivial}
\end{itemize}

Chechik et al.~\cite{chechik2010fault} study \textit{fault tolerant spanners} which were introduced originally for geometric graphs by Levcopoulos et al.~\cite{Levcopoulos1998}. Given a stretch parameter $t\geq1$ and a fault parameter $f\in\N$, an $f$--vertex (resp. $f$--edge) fault tolerant multiplicative $t$--spanner is a subgraph $G'=(V,E')$ such that for any set $F\subseteq V$ (resp. $F\subseteq E$) with $|F|\leq f$, we have
\[ d_{G'\setminus F}(u,v)\leq t\cdot d_{{\color{black}G\setminus F}}(u,v),\quad u,v\in V\setminus F\quad  (\textnormal{ resp. } u,v\in V).\]
The main result of~\cite{chechik2010fault} is to provide a randomized algorithm which computes an $f$--vertex fault tolerant $(2k-1)$--spanner with high probability containing $O(f^2k^{f+1}n^{1+\frac1k}\log^{1-\frac1k} n)$ edges in running time $O(f^2k^{f+2}n^{3+\frac1k}\log^{1-\frac1k}n)$.  For the edge fault tolerant case, their algorithm produces a $(2k-1)$--spanner with high probability containing $O(fn^{1+\frac1k})$ edges.

Ausiello et al.~\cite{ausiello2010computing} give improved constructions of fault tolerant spanners both in terms of time and number of edges compared to \cite{chechik2010fault}.  They construct an $f$--vertex fault tolerant $(3,2)$--spanner with $O(f^\frac43 n^\frac43)$ edges in $\widetilde{O}(f^2m)$ time, and an $f$--vertex fault tolerant $(2,1)$--spanner with $O(fn^\frac32)$ edges in $\widetilde{O}(fm)$ time.

Dinitz and Krauthgamer~\cite{dinitz2011fault} improve on the size of a fault tolerant spanner compared to Chechik et al.~\cite{chechik2010fault} by replacing the factor $k^{f+2}$ with $f^2$, thus giving a spanner polynomial in the fault size rather than exponential.  The main results of~\cite{dinitz2011fault} give a way to transform a spanner into a fault tolerant spanner. Specifically, if a $t$--spanner has $g(n)$ edges, then their algorithm produces an $f$--vertex fault tolerant $t$--spanner with $O(f^2\log n)g(\frac{2n}{f})$ edges.  Applying this to the greedy $(2k-1)$--spanner construction of~\cite{althofer1993sparse} yields an $f$--vertex fault tolerant $(2k-1)$--spanner with $\widetilde{O}(f^2n^{1+\frac{2}{k+1}})$ edges.  Additionally, in the unweighted case, they provide a $O(\log n)$--approximation algorithm for finding the minimal cost $f$ vertex fault $2$--spanner problem.

Bodwin et al.~\cite{bodwin2018optimal} give improved bounds on the size of fault tolerant spanners created by the direct analogue of the greedy algorithm of Alth\"{o}fer et al.~for fault tolerant spanners.
They show that for any positive integer $f$, there is an $f$--vertex (resp. $f$--edge) fault tolerant $(2k-1)$--spanner with $O(f^{1-\frac1k}n^{1+\frac1k} \exp(k))$ edges.
They also demonstrate that this bound is tight, up to the $\exp(k)$ factor, in the setting of vertex faults assuming Erd\H{o}s' Girth Conjecture.  Subsequently, Bodwin and Patel~\cite{bodwin2018trivial} greatly simplify the analysis of the so-called Fault Tolerant Greedy Algorithm, removing the $\exp(k)$ factor in the upper bound of~\cite{bodwin2018optimal}, and proving existential optimality (up to constant factors) even if the Girth Conjecture fails.

Braunschvig et al.~\cite{braunschvig2015fault} study additive and linear fault tolerant spanners.  They use a novel technique to combine a $(1,\beta)$--spanner construction with an $(\alpha,0)$--spanner construction to obtain a vertex and edge fault tolerant $(\alpha,\beta$)--spanner.  In particular, they prove that given an algorithm which produces an $(\alpha,\beta)$--spanner of size $O(n^{1+\delta})$, then for any $\eps>0$, one can obtain an $f_e$--edge, $f_v$--vertex fault tolerant $(\alpha+\eps,\beta)$--spanner of size $O((f_e+f_v)(\frac{\beta}{\eps})^{f_e+f_v}n^{1+\delta}\log n)$ provided $\max\{f_e,f_v\}<\lceil \frac{\beta}{\eps}\rceil +1$.  As a corollary, they use existing spanner constructions of Elkin and Peleg~\cite{elkin2004} to give a randomized construction of an $f$--edge, $f$--vertex fault tolerant $(1+\eps,\beta)$--spanner of size $O(f\beta(\frac{\beta}{\eps})^{2f}n^{1+\frac1k}\log n))$ with high probability.  Additionally, applying their result to the construction of Baswana et al.~\cite{baswana2010additive} gives an $f$--edge, $f$--vertex fault tolerant $(k+\eps,k-1)$--spanner with $O(f(\frac{k-1}{\eps})^{2f}n^{1+\frac1k}\log n)$ edges with high probability.

Bil\`{o} et al.~\cite{bilo2015improved} improve the constructions of other additive fault tolerant spanners in certain cases.  For $f=1$, they produce edge fault tolerant additive $2$--spanners of size $O(n^\frac53)$, $4$--spanners of size $O(n^\frac32)$, $10$--spanners of size $\widetilde{O}(n^\frac75)$ (with high probability), and $14$--spanners of size $O(n^\frac43)$.

Parter~\cite{parter2017fault} studies clustering + path buying algorithms for producing vertex fault tolerant additive spanners and sourcewise spanners for a single fault vertex.  The main results are to produce an additive $2$--spanner of size $O(n^\frac53)$, a $4$--additive sourcewise spanner of size $O(n|S|+(\frac{n}{|S|})^3)$, an additive $6$--spanner of size $O(n^\frac32)$, and an $8$--additive sourcewise spanner of size $O(n^\frac43)$ provided $|S|=O(n^\frac13)$.

Just as (non-faulty) BFS trees are frequently useful towards building non-faulty spanners, a corresponding notion of \emph{fault tolerant BFS structures (FTBFS)} is useful in many of the above constructions of fault tolerant spanners.
An $f$--edge or --vertex FTBFS is defined as an $S \times V$ distance preserver resilient to $f$ edge or vertex faults, in the same sense as the above.
These were introduced by Parter and Peleg \cite{PP13}, who showed that $O(|S|^{\frac12} n^{\frac32})$ edges are needed when $f=1$, for edge or vertex faults, and this is tight in either setting.
Subsequently, Parter \cite{Parter15} and Gupta and Khan \cite{GK17} proved tight bounds of $O(|S|^{\frac13} n^{\frac53})$ for either setting with $f=2$.
For general $f$, however, there remains a gap: Parter \cite{Parter15} proved a general lower bound of
$$\Omega(|S|^\frac{1}{f+1} n^{1 - \frac{1}{f+1}})$$
for $f$ edge or vertex faults, but the current best upper bound for general $f$ is
$$O(|S|^{\frac{1}{2^f}} n^{1 - \frac{1}{2^f}})$$
by \cite{BGPV17}.
It is a significant open question in the area to close this gap.
There is also a related area of fault tolerant \emph{reachability} trees, which must preserve reachability between all pairs in $S \times V$, not distance.

\subsection{Resilient Spanners}
\label{sec:ResilientSpanners}
\begin{itemize}
\setlength{\itemindent}{.2in}
\item[\cite{AusielloFIR16}] \bibentry{AusielloFIR16}
\end{itemize}

Ausiello et al.~\cite{AusielloFIR16} introduce the notion of \textit{resilience} in graph spanners.  Informally, a spanner is said to be \textit{resilient} if the stretch factor is not increased much by deleting an edge from a spanner.  Formally, if $G$ is any graph, then the \textit{fragility} of an edge $e$ is defined by
\[ \frag_G(e):=\max_{u,v\in V}\frac{d_{G\setminus e}(u,v)}{d_G(x,y)}.\]
Given a graph $G$, $\sigma\geq1$, $t\geq1$, and a $t$--spanner $G'$, an edge $e$ is $\sigma$--\textit{fragile} in $G'$ if 
$ \frag_{G'}(e)>\max\{\sigma,\frag_G(e)\},$
and the $t$--spanner $G'$ is $\sigma$--\textit{resilient} if 
every edge is \textit{not} $\sigma$--fragile.  That is, $G'$ is $\sigma$--resilient provided
\[\frag_{G'}(e) \leq \max\{\sigma, \frag_{G}(e)\},\quad e \in G'.\]

Resilience is a strong property, as the authors show that there is an infinite family of dense graphs which do not admit any $2$--resilient spanners other than the graphs themselves.  Moreover, resilience is a stronger notion than fault tolerance, and the authors show that there are $1$--edge fault tolerant $t$--spanners containing edges with fragility at least $\frac{t^2}{2}$ for any $t\geq3$.  Additionally, a polynomial time algorithm is given for producing a $\sigma$--resilient $(2k-1)$--spanner (for $\sigma\geq 2k-1$ and $k\geq2$) with $O(Wn^\frac32)$ edges, where $W=\frac{w_{\max}}{w_{\min}}$ with $w_{\max}$ being the largest weighted edge of $G$ and $w_{\min}$ being the smallest weighted edge. The runtime for this algorithm is $O(mn+n^2\log n)$. The authors also demonstrate that $(\alpha,\beta)$--spanners can be turned into $\sigma$--resilient spanners for any $\sigma\geq\alpha+\beta$.


\subsection{Dynamic Algorithms}
\label{sec:DynamicAlgorithms}

\begin{itemize}
\setlength{\itemindent}{.2in}
\item[\cite{AusielloFI06}] \bibentry{AusielloFI06}
\item[\cite{Baswana06ESA}] \bibentry{Baswana06ESA}
\item[\cite{baswana2008fully}] \bibentry{baswana2008fully}
\item[\cite{bodwin2016fully}] \bibentry{bodwin2016fully}
\item[\cite{BFH19}] \bibentry{BFH19}
\end{itemize}

Given a spanner for a graph, dynamic algorithms attempt to maintain the properties of the spanner while edges are being added to or deleted from the initial graph.  Thus, edges may need to be added or deleted in the spanner to maintain the distortion property.
Dynamic spanners have so far been studied with multiplicative error; it is unclear at present whether this is coincidental or if there is a hardness barrier to obtaining other types of error.

The initial work on the problem was in \cite{AusielloFI06}, where the authors present algorithms for maintaining a $3$-- or $5$--spanner of an input graph with essentially optimal size and update time proportional to the maximum degree.
This update time is \emph{amortized}, meaning that it holds on average over a sequence of insertions and deletions, but individual updates could take much longer.
In~\cite{baswana2008fully}, the authors present two algorithms for maintaining a sparse (multiplicative) $t$--spanner of an unweighted graph, again with optimal size (assuming the Girth Conjecture). The first algorithm achieves $O(7^\frac{t}{4})$ amortized update time (independent of the size of the graph $n$), and the second achieves $O(\textnormal{polylog} n)$ amortized update time (independent of the stretch factor $t$).

Bodwin and Krinninger~\cite{bodwin2016fully} address the problem of improving from \emph{amortized} to \emph{worst-case} update time, meaning that every individual update runs within the stated time with high probability.
They provide randomized algorithms to maintain a 3--spanner with $\widetilde{O}(n^{1+\frac{1}{2}})$ edges with worst-case update time $\widetilde{O}(n^\frac{3}{4})$, or a 5--spanner with $\widetilde{O}(n^{1+\frac{1}{3}})$ edges with worst-case update time $\widetilde{O}(n^\frac{5}{9})$.
Subsequently, \cite{BFH19} improved on these results by essentially converting the amortized construction of \cite{baswana2008fully} to worst-case update time, with only minor changes to the construction parameters.

A notable open question in the area is to progress from \emph{oblivious} to \emph{non-oblivious} update time.
That is: the update times in \cite{bodwin2016fully} and \cite{BFH19} hold with high probability against the randomness used in the algorithms, but only if the adversary choosing the graph updates is not allowed to see the random bits chosen by the algorithm.
A \emph{non-oblivious} construction would hold with high probability even if the adversary can base their updates on the random choices made by the construction.
This is an important property because, if dynamic spanners are used as a subroutine in other graph algorithms, the next graph update may depend on the current state of the spanner, which thus requires non-obliviousness to keep the guarantees.
The next step beyond non-obliviousness, of course, would be to obtain fully deterministic algorithms that maintain these spanners.

\section{Spanners for Special Classes of Graphs}\label{SEC:RestrictedGraphs}

While the sparse/light spanner problem is typically stated for generic graphs, it has also been studied when the class of input graphs is restricted. In many cases, much stronger guarantees can be made for spanners.  Here we highlight some of the literature in this vein, but we highlight the results rather than the techniques except where appropriate.

\subsection{Geometric Spanners}
\begin{itemize}
\setlength{\itemindent}{.2in}
\item[\cite{Narasimhan:2007:GSN:1208237}] \bibentry{Narasimhan:2007:GSN:1208237}
\end{itemize}

The book~\cite{Narasimhan:2007:GSN:1208237} provides an extensive treatment of geometric spanners, so we do not cover them in this survey.

\subsection{Directed Graphs}
\begin{itemize}
\setlength{\itemindent}{.2in}
\item[\cite{roditty2002roundtrip}] \bibentry{roditty2002roundtrip}
\item[\cite{dinitz2011directed}] \bibentry{dinitz2011directed}
\item[\cite{berman2011improved}] \bibentry{berman2011improved}
\item[\cite{dinitz2016approximating}] \bibentry{dinitz2016approximating}
\item[\cite{zhu2017source}] \bibentry{zhu2017source}
\item[\cite{zhu2018deterministic}] \bibentry{zhu2018deterministic}
\end{itemize}

Dinitz and Krauthgamer study multiplicative spanners for directed graphs, which is more subtle than the undirected case as one must maintain connectivity of the graph when computing a spanner. {\color{black}Dinitz et al.~\cite{dinitz2011directed}} propose a flow based linear program formulation of the directed $t$--spanner problem, and give an approximation algorithm to find sparsest $t$--spanner for a given directed graph. {\color{black}Berman et al.~\cite{berman2011improved} give a non-flow based formulation (an anti-spanner formulation) and improved the approximation of~\cite{dinitz2011directed}. Dinitz et al.~\cite{dinitz2016approximating} provide a $\tilde{O}(n^{1/3})$--approximation for directed $4$--spanner using a new rounding algorithm for the standard flow-based linear program.} Roditty et al.~\cite{roditty2002roundtrip} study roundtrip spanners, and Zhu and Lam \cite{zhu2017source,zhu2018deterministic} introduce the notion of sourcewise roundtrip spanners for directed graphs.

\subsection{Further Reading}

\begin{itemize}
\setlength{\itemindent}{.2in}
\item[\cite{Cohen00}] \bibentry{Cohen00}
\item[\cite{borradaile2017minor}] \bibentry{borradaile2017minor}
\end{itemize}


\section{Applications of Graph Spanners}\label{SEC:Applications}

Lastly, we survey some of the theoretical and practical applications of spanners in various fields.




\noindent \paragraph{Distributed Computing} A central challenge in distributed computing is message passing between vertices or nodes in an efficient manner.
Naively, nodes could constantly broadcast everything they know to all their neighbors -- thus propagating information around the system fairly quickly -- but this requires high information throughput and processing.
Often, a better idea is to build a sparse spanner of the network, greatly reducing the demands of sharing information at the price of only minor latency in propagation.
This paradigm appears, for example, in communication networks in parallel computing~\cite{Bhatt1986Optimal}, synchronizers~\cite{awerbuch1985complexity}, broadcasting~\cite{Peleg2000Distributed}, arrow distributed queuing protocols~\cite{Busch:2010:CCH:1866465.1866565}, wireless sensor networks~\cite{5061918}, online load balancing~\cite{awerbuch1992online}, and motion planning in robotics control optimization~\cite{DBLP:journals/ijcga/CaiK97}.
In all of these applications, the quality of the spanner that can be built controls the above tradeoff between latency and communication complexity.

\noindent \paragraph{Network Routing} Another related class of applications arises in the task of passing messages throughout a network (this is related to distributed computing in some ways, but often generalized or abstracted differently in the literature).
A classic example is where one wishes to pass packets or messages around the internet in a timely manner, but these ``messages'' can generally be construed quite broadly, e.g., as cars in a road network.
The routing challenge is more involved than simply computing a spanner, as solutions must balance the efficiency of the chosen paths with the amount of information stored at each node and in the ``packet header'' being passed around the network.
However, many modern routing algorithms exploit spanners as a useful step along the way.
See ~\cite{abraham2018ramsey, roditty2002roundtrip, TZ01, censor2018sparsest,elkin2014light} and references within for further information.

\noindent \paragraph{Computational Biology} To understand and model the history of organisms, Biologists have developed various ways to measure similarity between species, based either on their DNA sequence or on their level of interaction in an environment.
But given a matrix of pairwise similarities, how can these be arranged into a graph that succinctly captures biological history?
It turns out a good approach is to treat the matrix as a weighted graph (with all possible edges), and then build a spanner of the graph to determine its ``most important'' connections.
This application has been used, for example, to measure genetic distance between contemporary species~\cite{BANDELT1986309} and to visualize interactions between various proteins.~\cite{russel2005exploring}

\noindent \paragraph{Theoretical Applications} Since many graph algorithms have a runtime dependence on the number of edges in the input graph, one can often preprocess an input graph into a spanner in order to improve speed at the cost of a little accuracy.
One area where this has been applied is in polynomial time approximation schemes (PTAS) for the traveling salesman problem~\cite{borradaile2019greedy}.
Since improved spanners are often available on special graph classes, so too the PTAS can be improved.
Some of these graph classes include planar graphs, bounded-genus graphs, unit disk graphs, and bounded path width graphs~\cite{elkin2014light}.

Another class of theoretical applications show that spanners efficiently capture the ``backbone'' of the network and thus often implicitly represent other important network properties besides distances~\cite{alvarez2018mixed}.
This frequently includes the network spectrum, and accordingly spanners have been used to construct spectral sparsifiers~\cite{kapralov2014spanners}.

Other miscellaneous applications of spanners include computing distances and shortest paths between points embedding in a geometric space~\cite{bollobas2005sparse, chandra1992new}, testing graph properties (approximately) in sublinear time~\cite{chlamtac2016}, the facility location problem~\cite{choudhary2018diameter}, and key management in access control hierarchies~\cite{berman2011improved}.
Closer applications of spanners themselves include cycle covers of graphs~\cite{cai1995tree}, for which the extremal instances can often be decomposed into a union of tree spanners, and \emph{labelling schemes} in which the vertices of a graph are labelled in such a way that one can (approximately) recover the distance between nodes by inspecting only their labels~\cite{baswana2010additive}.



\section{Conclusion}\label{SEC:Conclusion}
Since their advent, graph spanners have become an important object of study and have found a broad range of applications as well as a rich theory.  The goal of this survey was to introduce readers to the overarching techniques that have been employed to compute various types of spanners, and to tabulate the state-of-the-art algorithmic bounds in an accessible way.   Additionally, we have posed several open problems along the way which may be of interest to experts and non-experts alike. The literature on spanners continues to grow at a rapid pace, and is unlikely to stop in the near future.  Nonetheless, it is hoped that this survey will provide a guiding reference for the state of the field for some time. 

{\color{black}
\section{Acknowledgments}
We would like to thank the anonymous reviewers for their valuable feedback which substantially improved the presentation of the material in this survey.  Additionally, we take pleasure in thanking Michael Elkin for helpful comments, especially related to Sections \ref{SEC:LowerBounds} and \ref{SEC:Distributed}.}

\section{References}
\bibliography{references}

\newpage\appendix
\section{Tables of Spanner Guarantees}

\begin{table}[H]
     \centering
     \begin{tabular}{|c|c|c|c|c|}
     \hline
         Stretch ($t$)  & Size: $O(|E'|)$ & Weight: $O(W(E'))$ & Time $O(\cdot)$ & Reference  \\
         \hline 
         \hline
         $3$ & $O(n^{\frac23})$ & ** & $\log n$ & \cite{derbel2007deterministic} \\ \hline
          $2k-1$  & $n^{1+\frac1k}$ & $(1+\frac{n}{2k})$ & $m(n^{1+\frac1k}+n\log n)$ &~\cite{althofer1993sparse}\\ 
         \hline
          $2k-1$  & $n^{1+\frac1k}$ & $kn^{1+\frac1k}$ & $km$ & \cite{Surender03ALP, baswana2007simple} \\ 
         \hline
         $2k-1$ & $O(kn^{1+\frac1k})$ & ** & ** & \cite{derbel2008locality} \\ 
         \hline
         $2k-1$ & $O(\min(m,kn^{1+\frac1k}))$  & ** & m & \cite{baswana2008streaming} \\ 
         \hline
         $(2k-1)(1+\eps)$  & $n^{1+\frac1k}$ & $kn^\frac1k \left(\frac1\eps\right)^{1+\frac1k}$ & $m(n^{1+\frac1k}+n\log n)$ &~\cite{chandra1992new}\\
         \hline
          $(2k-1)(1+\eps)$  & $n^{1+\frac1k}$ & $kn^\frac1k \left(\frac1\eps\right)^{1+\frac1k}$ & $n^\frac1k\left(1+\frac{k}{\eps^{1+\frac1k}\log k}\right)$ &~\cite{elkin2014light}\\
          \hline
          $(2k-1)(1+\eps)$  & $n^{1+\frac1k}\left(\frac1\eps\right)^{2+\frac1k}$ & $kn^\frac1k\left(\frac1\eps\right)^{2+\frac1k}$ & 
          $kn^{2+\frac1k}$ &~\cite{elkin2016fast}\\ 
          \hline
           $(2k-1)(1+\eps)$  & $n^{1+\frac1k}\left(k+\frac1\eps\right)^{2+\frac1k}$ & $kn^\frac1k\left(\frac1\eps\right)^{2+\frac1k}$ & 
          $km+\min\{n\log n, m\alpha(n)\}$ &~\cite{elkin2016fast}\\
         \hline
          $(2k-1)(1+\eps)$  & $n^{1+\frac1k}$ & $n^{1+\frac1k}\left(\frac1\eps\right)^{2+\frac1k}$ & ** &~\cite{chechik2018near}\\
          \hline
     \end{tabular}
     \caption{Guarantees for Multiplicative Spanners. Here, $n=|V|$, $m=|E|$, $\alpha(n)$ is the inverse Ackermann function, and all weight bounds are multiplied by $W(\textnormal{MST}(G))$.  Complexities not given in the paper are denoted by **. }
     \label{TAB:MultiplicativeSpanners}
 \end{table}
 
 \begin{table}[H]
     \centering
     \begin{tabular}{|c|c|c|c|c|}
     \hline
         Additive Error ($\beta$)  & Size: $O(|E'|)$ &  Time $O(\cdot)$ & Reference  \\
         \hline\hline
          2 & $\widetilde{O}(n^{\frac{3}{2}})$ &  $O(n^\frac52\sqrt{ \log n})$ & \cite{Aingworth99fast}\\
          \hline
          2 & $n^\frac{3}{2}$ &  $n^\frac52$ & \cite{Aingworth99fast}\\
          \hline
          2 & $n^\frac{3}{2}\log^\frac12n$ &  $n^2\log^2n$ & \cite{dor2000all}\\
          \hline
          2 & $n^\frac{3}{2}$ &  $n^2$ & \cite{knudsen2017additive}\\ 
         \hline
         4 & $O(n^{\frac75} \log^\frac{4}{5}(n)  )$  & $n^\frac{3}{5}\log^\frac{1}{5}n+D$ & \cite{censor2016distributed} \\ 
          \hline
          4 & $\widetilde{O}(n^\frac{7}{5})$ &  ** & \cite{chechik2013new}\\
          \hline
          6 & $n^\frac{4}{3}$ &  $mn^\frac{2}{3}$ & \cite{baswana2010additive}\\
         \hline
         6 & $n^\frac{4}{3}\log^3 n$ &  $n^2\log^2n$ & \cite{woodruff2010additive}\\
         \hline
         8 & $n^\frac{4}{3}$ &  $n^2$ & \cite{woodruff2010additive}\\
         \hline
         $(1,2k+4\ell)$ & $\Gamma_k(G)+n^{1+\frac{1}{k+\ell+1}}$ & $mn^{1-\frac{l}{k+l+1}}$ & \cite{baswana2010additive}\\ 
          \hline
         $(1,O( \sqrt{ d(u,v)}))$ & $\widetilde{O}(n^{1+\frac{3}{17}})$  & ** & \cite{chechik2013new} \\
         \hline
     \end{tabular}
     \caption{Guarantees for Additive Spanners. Here, $n=|V|$, $m=|E|$.  Complexities not given in the paper are denoted by **.  The term $D$ is the diameter of the graph. }
     \label{TAB:AdditiveSpanners}
 \end{table}
 
 \begin{table}[H]
     \centering
     \begin{tabular}{|c|c|c|c|}
     \hline
         ($\alpha,\beta$)  & Size: $O(|E'|)$ & Time $O(\cdot)$ & Reference  \\
         \hline\hline
          $(k - 1, 2k - O(1))$ & $kn^{1+\frac{1}{k}}$ & $mn^{1-\frac{1}{k}}$ & \cite{elkin2004}\\
         \hline
          $(k, k - 1)$ & $kn^{1+\frac{1}{k}}$ & $km$ & \cite{baswana2010additive} \\
         \hline
          $(1 + \epsilon, 4)$ & $\epsilon^{-1}n^\frac{4}{3}$ & $mn^\frac{2}{3}$ & \cite{elkin2004} \\
         \hline
         $(1 + \epsilon, \beta)$ & $\beta n^{1+\frac{1}{k}}$ & $n^{
2+\frac{1}{t}}$ & \cite{elkin2004}, $\beta = \beta(k, \epsilon, t)$ \\
         \hline
          $(1 + \epsilon, \beta')$ & $\beta'n^{1+\frac{1}{k}}$ & $mn^{\rho}$ & \cite{elkin2005computing}, $\beta' = \beta'(k, \epsilon, \rho)$ \\
         \hline
          $(1 + \epsilon, \beta'')$ & $kn^{1+\frac{1}{k}}$ & $kmn^{\frac{1}{k}}$ & \cite{thorup2006spanners}, $\beta'' = \beta''(k, \epsilon)$ \\
         \hline
     \end{tabular}
     \caption{Guarantees for $(\alpha,\beta)$--spanners. Here, $n=|V|$, $m=|E|$. In \cite{elkin2005computing}, $\rho > \frac{1}{2k}$ is required. }
     \label{TAB:AlphaBetaSpanners}
 \end{table}
 
 \begin{table}[H]
     \centering
     \begin{tabular}{|c|c|c|c|c|}
     \hline
         Spanner & ($\alpha,\beta$)  & Size: $O(|E'|)$ & Reference  \\
         \hline\hline
         Pairwise & $(1+\varepsilon,4)$ & $n|P|^\frac14\sqrt{\log\frac{n}{\eps}}$ &  \cite{cygan13} \\ \hline
         Pairwise & $(1,4k)$ & $n^{1+\frac{1}{2k+1}}((4k+5)|P|)^\frac{k}{4k+2}$ & \cite{cygan13}\\ \hline
         Pairwise & $(1,2)$ & $n|P|^\frac14$ &  \cite{kavitha2017new} \\
         \hline
          Pairwise & $(1, 2)$ & $O(n|P|^\frac{1}{3})$ & \cite{abboud2016error}\\
         \hline
         Pairwise & $(1,2)$ & $\widetilde{O}(n|P|^\frac13)$ & \cite{kavitha2013small} \\ 
         \hline
         Pairwise  & $(1,2)$ & $O(n |P|^{\frac13} \log^{\frac{2}{3}}(n)  )$  & \cite{censor2016distributed} \\
         \hline
          Pairwise & $(1, 4)$ & $\widetilde{O}(n|P|^\frac{2}{7})$ & \cite{kavitha2017new} \\
         \hline
          Pairwise & $(1, 6)$ & $n|P|^\frac{1}{4}$ & \cite{kavitha2017new} \\
         \hline
          Sourcewise & $(1, 2)$ & $\widetilde{O}(n(n|S|)^\frac{1}{4})$ & \cite{kavitha2013small} \\
         \hline
         Sourcewise  & $(1,2)$ & $O(n^{\frac54} |S|^{\frac14} \log^{\frac{3}{4}}(n)  )$  & \cite{censor2016distributed} \\
         \hline
         Sourcewise & $(1,2)$ & $\widetilde{O}(n(n|S|)^\frac14)$ & \cite{kavitha2013small} \\
         \hline
          Sourcewise & $(1, 4)$ & $\widetilde{O}(n(n|S|)^\frac{2}{9})$ & \cite{kavitha2017new} \\
         \hline
          Sourcewise & $(1, 6)$ & $\widetilde{O}(n(n|S|)^\frac{1}{5})$ & \cite{kavitha2017new} \\
         \hline
          Sourcewise & $(1, 2k), \forall k \geq 1$ & $\widetilde{O}(n(n|S|^k)^{1/(2k+2)})$ & \cite{Parter14} \\
         \hline
         Sourcewise $(S\times V)$ & $(1,2k)$ & $n^{1+\frac{1}{2k+1}}(k|S|)^\frac{k}{2k+1}$ &  \cite{cygan13}\\
         \hline
          Subsetwise & $(1, 2)$ & $n\sqrt{|S|}$ & \cite{cygan13,elkincomm,pettie2009low} \\
         \hline 
         Subsetwise  & $(1,2)$ & $O(n |S|^{\frac23} \log^{\frac{2}{3}}(n)  )$  & \cite{censor2016distributed} \\
         \hline
          Source-Target & $(1, 4)$ & $n(|S||T|)^\frac{1}{4}$ & \cite{kavitha2017new} \\
         \hline
          Source-Target & $(1, 2k), \forall k \geq 1$ & $n(|S|^{k+1}|T|)^{1/(2k+3)}$ & \cite{kavitha2017new} \\
         \hline
     \end{tabular}
     \caption{Guarantees for pairwise and subsetwise spanners. In~\cite{Klein06}, an algorithm has been provided to compute multiplicative subsetwise spanner with constant approximation ratio for planar graphs.}
     \label{TAB:Pairwise}
 \end{table}
 
  \begin{table}[H]
     \centering
     \begin{tabular}{|c|c|c|c|c|c|}
     \hline
         Fault ($f$, E/V) & $(\alpha,\beta)$ & Size: $O(|E'|)$ & Time $O(\cdot)$ & Reference  \\
         \hline\hline
           $f$, V & $(2k-1,0)$  & $f^2k^{f+1}n^{1+\frac1k}\log^{1-\frac1k}n$ &  $f^2k^{f+2}n^{3+\frac1k}\log^{1-\frac1k}n$ &~\cite{chechik2010fault}$^*$\\
         \hline
         $f$, E & $(2k-1,0)$  & $fn^{1+\frac1k}$ &  $f^2k^{f+2}n^{3+\frac1k}\log^{1-\frac1k}n$ &~\cite{chechik2010fault}$^*$\\
         \hline
         $f$, V & $(2k-1,0)$  & $f^{2-\frac1k}n^{1+\frac1k}\log n$ &  ** &~\cite{dinitz2011fault}\\
         \hline
         $f$, V & $(2k-1,0)$  & $f^{1-\frac1k}n^{1+\frac1k}2^{O(k)}$ &  ** &~\cite{bodwin2018optimal}\\
          \hline
          $f$, V & $(2k-1,0)$  & $f^{1-\frac1k}n^{1+\frac1k}$ &  ** &~\cite{bodwin2018trivial}\\
          \hline
          $f$, V & $(3,2)$  & $f^\frac43 n^\frac43 $ &  $\widetilde{O}(f^2m)$ &~\cite{ausiello2010computing}\\
          \hline
          $f$, V & $(2,1)$  & $fn^\frac32 $ &  $\widetilde{O}(fm)$ &~\cite{ausiello2010computing}\\
          \hline
          $f$, E, V & $(1+\eps,\beta)$  & $f\beta(\frac\beta\eps)^{2f}n^{1+\frac1k}\log n $ &  ** &~\cite{braunschvig2015fault}$^*$\\
          \hline
          $f$, E, V & $(k+\eps,k-1)$  & $f(\frac{k-1}{\eps})^{2f}n^{1+\frac1k}\log n $ &  ** &~\cite{braunschvig2015fault}$^*$\\
          \hline
           $1$, E & $(1,2)$  & $n^\frac53 $ &  ** &~\cite{bilo2015improved}\\
          \hline
          $1$, E & $(1,4)$  & $n^\frac32 $ &  ** &~\cite{bilo2015improved}\\
          $1$, E & $(1,10)$  & $\widetilde{O}(n^\frac75) $ &  ** &~\cite{bilo2015improved}$^*$\\
          \hline
          $1$, E & $(1,14)$  & $n^\frac43 $ &  ** &~\cite{bilo2015improved}\\
          \hline
          $1$, V & $(1,2)$  & $n^\frac53 $ &  ** &~\cite{parter2017fault}\\
          \hline
          $1$, V & $(1,6)$  & $n^\frac32 $ &  ** &~\cite{parter2017fault}\\
          \hline
          $f$, E, V & $(2k-1,0)$  & $f^{1-\frac1k}n^{1+\frac1k}2^{O(k)} $ &  ** &~\cite{bodwin2018optimal}\\
          \hline
          $f$, E, V & $(2k-1,0)$  & $f^{1-\frac1k}n^{1+\frac1k} $ &  ** &~\cite{bodwin2018trivial}\\
          \hline
     \end{tabular}
     \caption{Guarantees for Fault Tolerant Spanners. Here, $n=|V|$, $m=|E|$, $f$ is the number of faults allowed, and E or V in the first column denotes if the guarantee is for edge-fault or vertex-fault tolerance. Complexities not given in the paper are denoted by **.  Spanners constructed with high probability are denoted by $^*$ next to the citation.}
     \label{TAB:FaultTolerant}
 \end{table}

\end{document}